\documentclass[envcountsec]{svmult}
\usepackage{}
\usepackage{amsmath,amssymb}
\usepackage{graphicx}
\usepackage{hyperref}
\usepackage[numbers]{natbib}
\usepackage{enumerate}

\newcommand{\Grad}{\nabla}
\newcommand{\Div}{\nabla \cdot}

\newcommand{\Md}{\partial}

\newcommand{\Ga}{\alpha}

\newcommand{\Gr}{\rho}

\newcommand{\Go}{\omega}

\newcommand{\GO}{\Omega}

\newcommand{\bfm}[1]{\mbox{\boldmath ${#1}$}}

\newcommand{\BGve}{\bfm\varepsilon}

\newcommand{\BGm}{\bfm\mu}

\newcommand{\BGr}{\bfm\rho}

\newcommand{\BGs}{\bfm\sigma}

\newcommand{\CI}{{\cal I}}

\newcommand{\CO}{{\cal O}}

\def\Ba{{\bf a}}

\def\Bc{{\bf c}}

\def\Bk{{\bf k}}

\def\Bn{{\bf n}}

\def\Bp{{\bf p}}

\def\Bt{{\bf t}}
\def\Bu{{\bf u}}
\def\Bv{{\bf v}}
\def\Bw{{\bf w}}
\def\Bx{{\bf x}}
\def\By{{\bf y}}
\def\Bz{{\bf z}}
\def\BA{{\bf A}}
\def\BB{{\bf B}}
\def\BC{{\bf C}}
\def\BD{{\bf D}}

\def\BF{{\bf F}}
\def\BG{{\bf G}}

\def\BI{{\bf I}}

\def\BR{{\bf R}}
\def\BS{{\bf S}}

\newcommand{\real}{\mathbb{R}}
\newcommand{\complex}{\mathbb{C}}

\newcommand{\tu}{\widetilde{u}}

\newcommand{\tU}{\widetilde{U}}
\newcommand{\tg}{\widetilde{g}}

\renewcommand{\hat}[1]{\widehat{#1}}

\newcommand{\abs}[1]{\left|{#1}\right|}

\newcommand{\M}[1]{\left({#1}\right)}
\newcommand{\Mb}[1]{\left[{#1}\right]}
\newcommand{\Mcb}[1]{\left\{{#1}\right\}}
\newcommand{\ceil}[1]{\lceil{#1}\rceil}
\newcommand{\conj}[1]{\overline{#1}}

\newcommand{\rlab}[2]{\raisebox{#2}{\rotatebox{90}{#1}}}

\smartqed

\newcommand{\secref}[1]{Sect.~\ref{#1}}
\newcommand{\figref}[1]{Fig.~\ref{#1}}
\newcommand{\thmref}[1]{Theorem~\ref{#1}}

\begin{document}
\title*{Transformation elastodynamics and active exterior acoustic cloaking}
\author{Fernando Guevara Vasquez, Graeme W. Milton, Daniel
Onofrei and Pierre~Seppecher}
\authorrunning{F. Guevara~Vasquez, G.W. Milton, D. Onofrei and P. Seppecher}
\institute{%
 Fernando Guevara Vasquez, \email{fguevara@math.utah.edu} \and 
 Graeme W. Milton, \email{milton@math.utah.edu} \and 
 Daniel Onofrei, \email{onofrei@math.utah.edu}, 
 \at Department of Mathematics, University of Utah, Salt Lake City, UT 84112, USA.
 \and
 Pierre Seppecher, \email{seppecher@imath.fr}, 
 \at Institut de Math\'ematiques de Toulon, Universit\'e de Toulon et du Var,
 BP 132-83957 La Garde Cedex, France.}
\maketitle

Coordinate transformations can be used to manipulate fields in a variety
of ways for the Maxwell and Helmholtz equations. In \secref{sec:ted} we
focus on transformation elastodynamics. The idea is to manipulate waves
in an elastic medium by designing appropriate transformations of the
coordinates and the displacements. As opposed to the Maxwell and
Helmholtz equations, the elastodynamic equations are not invariant under
these transformations. Here we recall the transformed elastodynamic
equations, and then move to the effect of space transformations on a
mass-spring network model. In order to realize the transformed networks
we introduce ``torque springs'', which are springs with a force
proportional to the displacement in a direction other than the direction
dictated by the spring terminals.  We discuss some possible
homogenizations of transformed networks that could have applications to
manipulating waves in an elastic medium for e.g. cloaking.

Then we look at an approach to cloaking which is based on cancelling the
incident field using active devices (rather than passive composite
materials) which are exterior to the cloaked region. Exterior means that
the cloaked region is not completely surrounded by the cloak, as is the
case in most transformation based methods. We present here active
exterior cloaking methods for both the Laplace equation in dimension two
(\secref{sec:static}) and the Helmholtz equation in dimension three
(\secref{sec:helm}). 

The cloaking method for the Laplace equation we present in \secref{sec:static}
applies also to the quasi-static (low frequency) regime and was in part
presented in \cite{Vasquez:2011:MAA,Vasquez:2009:AEC}. We first
reformulate the problem of designing an active cloaking device as the
classic problem of approximating analytic functions with polynomials.
This theoretical approach shows that it is possible to cloak an object
from an incident field with one single exterior device. Then we give an
explicit solution to the problem in terms of a polynomial and determine
its convergence region as the degree of the polynomial increases.  This
convergence region limits the size of the cloaked region, and for the
new solution we propose here it allows one to cloak larger objects at a
fixed distance from the device compared to the explicit polynomial
solution given in \cite{Vasquez:2011:MAA,Vasquez:2009:AEC}. We also
discuss how our approach can be modified to simultaneously hide an
object and give the illusion of another object, in the same spirit as
illusion optics \cite{Lai:2009:IOO}.

Next in \secref{sec:helm} we consider the Helmholtz equation and use the
same techniques as in \cite{Vasquez:2011:ECA} to show that in dimension
three it is possible to cloak an object using four devices and yet
leaving the object connected with the exterior. Our method is based on
Green's formula, which ensures that an analytic field can be reproduced
inside a volume by a carefully chosen single and double layer potential
at the surface of the volume. Then we use addition theorems for
spherical outgoing waves to concentrate the single and double layer
potential at a few multipolar sources (cloaking devices) located outside
the cloaked region. We determine the convergence region of the device's
field and include an explicit geometric construction of a cloak with
four devices.

The three sections of this chapter can be read essentially independently
of each other.

%%%%%%%%%%%%%%%%%%%%%%%%%%%%%%%%%%%%%%%%%%%%%%%%%%%%%%%%%%%%%%%%%%%%%%%
\section{Transformation elastodynamics}
\label{sec:ted}
Transformation based cloaking was first discovered by Greenleaf, Lassas
and Uhlmann \cite{Greenleaf:2003:ACC,Greenleaf:2003:NCI} in the context
of the conductivity equations. Independently, Leonhardt realized that
transformation based cloaking applies to geometric optics
\cite{Leonhardt:2006:OCM} and Pendry, Schurig and Smith
\cite{Pendry:2006:CEM} realized that transformation based cloaking
applies to Maxwell's equations at fixed frequency, and this led to an
explosion of interest in the field.  It was found that transformation
based cloaking also applies to acoustics
\cite{Cummer:2007:PAC,Chen:2007:ACT,Greenleaf:2007:FWI}, which is
governed by the Helmholtz equation, provided one permits anisotropic
density \cite{Schoenberg:1983:PPS}.  These developments, reviewed in
\cite{Alu:2008:PMC,Greenleaf:2009:CDE,Cai:2010:OM,Chen:2010:ACT} rely
on the invariance of the conductivity equations, Maxwell's equations,
and the Helmholtz equation under coordinate transformations, and have
been substantiated by rigorous proofs
\cite{Greenleaf:2007:FWI,Kohn:2008:CCV,Kohn:2010:CCV}.  The invariance
of Maxwell's equations under coordinate transformations has led to other
envisaged applications such as field concentrators \cite{Rahm:2008:DEC},
field rotators \cite{Chen:2009:DER}, lenses \cite{Schurig:2008:AFL},
superscatterers \cite{Yang:2008:SES} (see also
\cite{Nicorovici:1994:ODP}) and the name ``transformation optics'' is
now used to describe this research: see, for example, the special issue
in the New Journal of Physics \cite{Leonhardt:2008:FCT} devoted to
cloaking and transformation optics. The perfect lens of Pendry
\cite{Pendry:2000:NRM} can be viewed as the result of using a
transformation which unfolds space \cite{Leonhardt:2006:GRE} and
associated with such folding transformations is cloaking due to
anomalous resonance
\cite{Milton:2006:CEA,Nicorovici:2007:OCT,Milton:2008:SFG}.

A largely open question is how to construct metamaterials with the
required combination of anisotropic electrical permittivity $\BGve(\Bx)$
and anisotopic magnetic permeability $\BGm(\Bx)$ needed in
transformation optics designs, frequently with $\BGve(\Bx)=\BGm(\Bx)$.
Only recently was it shown \cite{Milton:2010:RMP}, building upon work of
Bouchitt\'e and Schweizer \cite{Bouchitte:2010:HME}, that any
combination of real tensors $(\BGve,\BGm)$ is approximately realizable,
at least in theory.

Curiously, the usual elastodynamic equations do not generally keep their
form under coordinate transformations. Either new terms enter the
equations \cite{Milton:2006:CEP}, so they take the form of equations
Willis introduced \cite{Willis:1981:VPDP} to describe the ensemble
averaged elastodynamic behavior of composite materials (which are the
analog of the bianisotropic equations of electromagnetism
\cite{Serdyukov:2001:EBAM}), or the elasticity tensor field does not
retain its minor symmetries \cite{Brun:2009:ACI}. Nevertheless, as shown
in \cite{Milton:2007:NMM} and as is explored further here, there is
some hope that metamaterials can be constructed with a response
corresponding approximately with that required by the new equations.  

%%%%%%%%%%%%%%%%%%%%%%%%%%%%%%%%%%%%%%%%%%%%%%%%%%%%%%%%%%%%%%%%%%%%%%%
\subsection{Continuous transformation elastodynamics}
By extending the analysis of \cite{Milton:2006:CEP}, let us show that
the equation of elastodynamics 
\begin{equation}
-\Div(\BC(\Bx)\Grad\Bu)=\Go^2\Gr(\Bx)\Bu 
\label{a.1}
\end{equation}
changes under the transformation
\begin{equation}
\Bx'=\Bx'(\Bx),\quad \Bu'(\Bx'(\Bx))=(\BB^T(\Bx))^{-1}\Bu(\Bx) 
\label{a.2}
\end{equation}
to the equation
\begin{equation}
-\nabla'\cdot(\BC'(\Bx')\Grad'\Bu'+\BS'(\Bx')\Bu')+\BD'(\Bx')\Grad'\Bu'-\Go^2(\BGr'(\Bx')\Bu')=0
\label{a.6a}
\end{equation}
where the tensors $\BC'$, $\BS'$, $\BD'$, $\BGr'$ are given in terms of
the functions $\Bx'$, $\BB$ and their derivatives. Here the
transformation of the displacement is governed by $\BB(\Bx)$ which can
be chosen to be any invertible matrix valued function. (The inverse and
transpose in $(\BB^T(\Bx))^{-1}$ have been introduced to simplify
subsequent formulae.) 

Indeed let us first note that
\begin{eqnarray} 
\Grad\Bu & = & \frac{\Md u_j}{\Md x_i}=\frac{\Md(u'_pB_{pj})}{\Md x_i} 
= \frac{\Md x_m'}{\Md x_i}\frac{\Md u'_p}{\Md x'_m}B_{pj}+ \frac{\Md
B_{pj}}{\Md x_i}u'_p \nonumber\\
& = & \BA^T(\Grad'\Bu')\BB+\BG'\Bu'
\label{a.3}
\end{eqnarray}
in which $\BA$ and $\BG$ are the tensors with elements
\begin{equation}
A_{mi}=\frac{\Md x_m'}{\Md x_i},\quad G_{ijp}=\frac{\Md B_{pj}}{\Md x_i}.
\label{a.4}
\end{equation}
Now \eqref{a.1} implies that for all smooth vector-valued test functions
$\Bv(\Bx)$ with compact support in a domain $\GO$,  
\begin{eqnarray}
& 0 & = \int_{\GO}[-\Div(\BC(\Bx)\Grad\Bu)-\Go^2\Gr(\Bx)\Bu]\cdot\Bv~\D\Bx \nonumber\\
& ~ & = \int_{\GO}[\BC(\Bx)\Grad\Bu:\Grad\Bv-\Go^2\Gr(\Bx)\Bu\cdot\Bv]~\D\Bx\nonumber\\
& ~ & = \int_{\GO'}[\BC(\Bx)(\BA^T(\Grad'\Bu')\BB+\BG\Bu'):(\BA^T(\Grad'\Bv')\BB+\BG\Bv')
-\Go^2\Gr(\Bx)(\BB^T\Bu')\cdot(\BB^T\Bv')]a^{-1}~\D\Bx'\nonumber\\
& ~ & =\int_{\GO'}[\BC'(\Bx')\Grad'\Bu':\Grad'\Bv'+\BS'(\Bx')\Bu':\Grad'\Bv'+(\BD'(\Bx')\Grad'\Bu')\cdot\Bv'
-\Go^2(\BGr'(\Bx')\Bu')\cdot\Bv']~\D\Bx'\nonumber\\
& ~ & =\int_{\GO'}[-\nabla'\cdot(\BC'(\Bx')\Grad'\Bu'+\BS'(\Bx')\Bu')+\BD'(\Bx')\Grad'\Bu'-\Go^2(\BGr'(\Bx')\Bu')]\cdot\Bv'~\D\Bx'\nonumber\\
\label{a.5}
\end{eqnarray}
in which the test function $\Bv(\Bx)$ has been transformed, similarly to
$\Bu(\Bx)$, to 
\begin{equation}
 \Bv'(\Bx'(\Bx))=(\BB^T(\Bx))^{-1}\Bv(\Bx),
\label{a.5a}
\end{equation}
and $a(\Bx'(\Bx))=\det \BA(\Bx)$ while $\BC'(\Bx')$, $\BS'(\Bx')$,
$\BD'(\Bx')$ and $\BGr'(\Bx')$ are the tensors with elements
\begin{eqnarray}
C'_{ijk\ell}& = & a^{-1}A_{ip}B_{jq}A_{kr}B_{\ell s}C_{pqrs},\nonumber\\
S'_{ijk} & = & a^{-1}A_{ip}B_{jq}G_{rsk}C_{pqrs}=a^{-1}A_{ip}B_{jq}\frac{\Md B_{ks}}{\Md x'_r}C_{pqrs},\nonumber\\
D'_{kij} & = & a^{-1}G_{pqk}A_{ir}B_{js}C_{pqrs}=S'_{ijk},\nonumber\\
\Gr'_{ij} & = & a^{-1}B_{ik}B_{jk}\Gr- a^{-1}\Go^{-2}G_{pqi}G_{rsj}C_{pqrs}\nonumber\\
& = &  a^{-1}B_{ik}B_{jk}\Gr
- a^{-1}\Go^{-2}\frac{\Md B_{iq}}{\Md x'_p}\frac{\Md B_{js}}{\Md x'_r}C_{pqrs}.
\label{a.6}
\end{eqnarray}
From \eqref{a.5} we see directly that \eqref{a.1} transforms to \eqref{a.6a}.

\begin{remark}
The transformed elastodynamic equation \eqref{a.6a} can be written in the
equivalent form of Willis-type equations \cite{Willis:1981:VPDP} 
\begin{eqnarray}
\nabla'\cdot\BGs' & = & -i\Go \Bp', \nonumber\\ \quad \BGs' & = &
\BC'(\Bx')\Grad'\Bu'+(i/\Go)\BS'(\Bx')(-i\Go\Bu'), \nonumber\\ \Bp'& = &
\BGr'(\Bx')(-i\Go\Bu')+(i/\Go)\BD'(\Bx')\Grad\Bu', 
\label{a.6b}
\end{eqnarray}
in which the stress $\BGs'$, which is not necessarily symmetric, depends
not only upon the dispacement gradient $\Grad'\Bu'$ but also upon the
velocity $-i\Go\Bu'$, and the momentum $\Bp'$ depends not only upon the
velocity $-i\Go\Bu'$, but also on the displacement gradient
$\Grad'\Bu'$.
\end{remark}

\begin{remark}
If we desire the transformed elasticity tensor $\BC'(\Bx')$ to have all
the usual symmetries of elasticity tensors, namely that
\begin{equation}
C'_{ijk\ell}=C'_{jik\ell}=C'_{k\ell ij}, 
\label{a.6c}
\end{equation}
then we need to restrict the transformations to those with $\BB=\BA$.
This was the case analyzed by Milton, Briane and Willis
\cite{Milton:2006:CEP}.
\end{remark}

\begin{remark}
In the particular case where $\BB=\BI$ the transformation \eqref{a.6}
reduces to
\begin{eqnarray} 
C'_{ijk\ell}& = & a^{-1}A_{ip}A_{kr}C_{pjr\ell}, \quad \BS'=\BD'=0,\quad \BGr'=a^{-1}\Gr\BI,
\label{a.7}
\end{eqnarray}
corresponding to normal elastodynamics, with an isotropic density matrix
$\BGr'$, but with an elasticity tensor $\BC'$ only satisfying the major
symmetry $C'_{ijk\ell}=C'_{k\ell ij}$. This was the case analysed by
Brun, Guenneau and Movchan \cite{Brun:2009:ACI} in a particular
two-dimensional example.
\end{remark}

Having derived the rules of transformation elasticity, one can then
apply the same variety of transformations as used in transformation
optics, including cloaking and folding transformations. The point is
that a wave propagating classically in the classical medium can have a
strange behavior in the new abstract coordinate system $\Bx'$. If we are
able to design a real medium following a system of equations equivalent
to the transformed system, then we are able to force a strange behavior
for waves in real physical space.

%%%%%%%%%%%%%%%%%%%%%%%%%%%%%%%%%%%%%%%%%%%%%%%%%%%%%%%%%%%%%%%%%%%%%%%
\subsection{Discrete transformation elastodynamics}
\label{sec:dte}
There is a discrete version of the transformation \eqref{a.7}. Suppose we
have a network of springs, possibly a lattice infinite in extent, with a
countable number of nodes at positions $\Bx_1, \Bx_2, \Bx_3, \ldots,
\Bx_n \ldots$, at which there are masses $M_1, M_2, M_3, \ldots,
M_n\ldots$, and at which the displacements are  $\Bu_1, \Bu_2, \Bu_3,
\ldots, \Bu_n \ldots$.  Let $k_{ij}$ denote the spring constant of the
spring connecting node $i$ to node $j$. There is no loss of generality
in assuming that all pairs of nodes are joined by a spring, taking
$k_{ij}=0$ if there is no real spring joining node $i$ and $j$.  Let
$\BF_{i,j}$ denote the force which the spring joining nodes $i$ and $j$
exerts on node $i$.  Hooke's law implies 
\begin{equation} 
\BF_{i,j}=-\BF_{j,i}=k_{i,j}\Bn_{i,j}[\Bn_{i,j}\cdot(\Bu_{j}-\Bu_i)], 
\label{a.8}
\end{equation}
where 
\begin{equation} 
\Bn_{i,j}=\frac{\Bx_j-\Bx_i}{|\Bx_j-\Bx_i|}, 
\label{a.9}
\end{equation}
is the unit vector in the direction of $\Bx_j-\Bx_i$.  In the absence of
any forces acting on the nodes, apart from inertial forces, Newton's
second law implies
\begin{equation} 
\sum_{j}\BF_{i,j}=-M_i\Go^2\Bu_i. 
\label{a.10}
\end{equation}
Now let us consider a transformation $\Bx'=\Bx'(\Bx)$ with an associated
inverse transformation $\Bx=\Bx(\Bx')$. Under this transformation the
position of the nodes transform to $\Bx'_1, \Bx'_2, \Bx'_3, \ldots,
\Bx'_n \ldots$, where $\Bx'_i=\Bx'(\Bx_i)$. We focus, for simplicity, on
the case corresponding to $\BB=\BI$ where the forces, masses and
displacements transform according to 
\begin{equation} 
\BF'_{i,j}=\BF_{i,j},\quad M'_i=M_i,\quad \Bu'_i=\Bu_i. 
\label{a.11}
\end{equation}
After the transformation, Newton's second law clearly keeps its form,
\begin{equation}  
\sum_{j}\BF'_{i,j}=-M'_i\Go^2\Bu'_i, 
\label{a.12}
\end{equation}
while \eqref{a.8} transforms to 
\begin{equation}
\BF'_{i,j}=-\BF'_{j,i}=k'_{i,j}\Bv'_{i,j}[\Bv'_{i,j}\cdot(\Bu'_{j}-\Bu'_i)] 
\label{a.13}
\end{equation}
where
\begin{equation} 
k'_{i,j}=k_{i,j},\quad \Bv'_{i,j}=\frac{\Bx(\Bx'_j)-\Bx(\Bx'_i)}{|\Bx(\Bx'_j)-\Bx(\Bx'_i)|}
\label{a.14}
\end{equation}

Hence in the new coordinates $\Bx'_i$ the system is governed by
equations similar to the classical  system of equations for a network of
masses joined by springs, but the response of the springs  does not
anymore correspond to normal springs.  While the action-reaction
principle $\BF'_{j,i}=-\BF'_{i,j}$ remains valid, the force $\BF'_{i,j}$
is not generally parallel to the line joining $\Bx'_j$ with $\Bx'_i$.

Now we desire to construct a real network having a behavior governed at
a fixed frequency, by the system of equations \eqref{a.12} and \eqref{a.13}.
To that aim we need to construct a two-terminal network made of
classical masses and springs which has the response \eqref{a.13} for any
unit vector $\Bv'_{i,j}$.  We call these two-terminal networks ``torque
springs'' since they extert a torque in addition to the usual spring
force. We show how they can be constructed for fixed frequency
$\omega$ in the next section.

%%%%%%%%%%%%%%%%%%%%%%%%%%%%%%%%%%%%%%%%%%%%%%%%%%%%%%%%%%%%%%%%%%%%%%%
\subsection{Torque springs}
A torque spring, being a two-terminal network with a response of the
type \eqref{a.13},  is characterized by two terminal nodes $\Bx_1$,
$\Bx_2$, the direction of exerted forces $\Bv_{1,2}$ which can be
different from the direction of the line joining $\Bx_1$ and $\Bx_2$ and
the constant of the spring $k_{1,2}$.  The existence of torque springs
is guaranteed by the work of Milton and Seppecher \cite{Milton:2008:RRM}
which provides a complete characterization of the response of
multiterminal mass-spring networks at a single frequency. The complete
characterization of the response of multiterminal mass-spring networks
as a function of frequency was subsequently obtained by Guevara~Vasquez,
Milton, and Onofrei \cite{Vasquez:2011:CCS}.  Here we are just
interested in constructing two terminal networks with the response of a
torque spring. In this case a simpler construction, than provided by the
previous work, is possible. 

Consider the network of \figref{fig:1}. For its design we start with
$\Bx_1$, $\Bx_2$ and a unit vector $\Bv=\Bv_{12}$  not parallel to
$\Bx_1-\Bx_2$.  (A normal spring can be used if $\Bv$ is parallel to
$\Bx_1-\Bx_2$.) Choose $\rho>0$ and define $\By_1=\Bx_1+ \rho \Bv$,
$\By_2=\Bx_2+ \rho \Bv$, and choose a vector $\Bw\ne 0$ in a direction
different from $\Bv$ and $\Bx_2-\Bx_1$. Define $\Bz_1=\By_1+ \Bw$,
$\Bz_2=\By_2+ \Bw$, $\Bt_1=\Bz_1+\Bv$, $\Bt_2=\Bz_2+\Bv$. The pairs
$(\Bx_1,\By_1)$, $(\Bx_2,\By_2)$, $(\By_1,\By_2)$, $(\By_1,\Bz_1)$,
$(\By_2,\Bz_2)$, $(\Bz_1,\Bz_2)$, $(\Bz_1,\Bt_1)$, $(\Bz_2,\Bt_2)$ are
joined with normal springs of constant $k$. Masses (with mass $m$, where
the lower case $m$ is used to identify them as internal masses of torque
springs) are attached to the nodes $\Bt_1$ and $\Bt_2$ only. All nodes
but $\Bx_1$ and $\Bx_2$ are interior nodes which means that no external
forces are exerted on them.

Let us denote by $T$ the tension in the spring $(\Bx_1,\By_1)$, taken to
be positive if the spring is under extension and negative if it is under
compression, i.e. the spring exerts a force $\Bv T$ on the terminal at
$\Bx_1$ and a force $-\Bv T$ on the node at $\By_1$. Then the balance of
forces at node $\By_1$ fixes the tensions $T'$, $T"$ in the springs
$(\By_1,\By_2)$, $(\By_1,\Bz_1)$ in a purely geometrical way.  It is
easy to check that the balance of  forces at $\By_2$, $\Bz_1$, and
$\Bz_2$ gives tensions $-T$, $-T"$, $-T'$, $T$, $-T$ in the springs
$(\Bx_2,\By_2)$, $(\By_2,\Bz_2)$, $(\Bz_1,\Bz_2)$, $(\Bz_1,\Bt_1)$,
$(\Bz_2,\Bt_2)$ respectively.

All the tensions being determined when one is known, the truss is rank
one : there is only one scalar linear combination of the displacements
$\Bu_1$, $\Bu_2$, $\Bw_1$, $\Bw_2$ of nodes $\Bx_1$, $\Bx_2$, $\Bt_1$,
$\Bt_2$ which influences $T$, and $T=0$ if and only if this scalar
linear combination vanishes. It is easy to check that this combination
is $(\Bu_2-\Bu_1-\Bw_2+\Bw_1)\cdot \Bv$ since displacements leaving this
zero (floppy modes) do not produce any tension in the springs, as they
leave the spring lengths invariant to first order in the displacements.
Hence there exists a constant $K$ (proportional to $k$) such that $T= K
(\Bu_2-\Bu_1-\Bw_2+\Bw_1)\cdot \Bv$. Finally Newton's law \eqref{a.10}
gives at nodes $\Bt_1$, $\Bt_2$ respectively $T=m \omega^2 \Bw_1\cdot
\Bv $ and $-T=m \omega^2 \Bw_2\cdot \Bv$ and so $T= K (\Bu_2-\Bu_1)\cdot
\Bv + 2 T K m^{-1}\omega^{-2}$ from which we conclude that
\begin{equation} 
T=\frac{K m\omega^2}{m \omega^2-2 K} (\Bu_2-\Bu_1)\cdot \Bv.
\label{a.15}
\end{equation}
The forces $\BF_1$ and $\BF_2$ which this torque spring exerts on
terminals 1 and 2, respectively, are therefore
\begin{equation} 
\BF_1=-\BF_2=T\Bv= k'\Bv[\Bv\cdot(\Bu_2-\Bu_1)],\quad {\rm
with~}k'=\frac{Km\Go^2}{m\Go^2-2K},
\label{a.17}
\end{equation}
which is exactly of the required form \eqref{a.13}. If we want $k'$ to be
positive then we should choose $m$ and $K$ so that $m\Go^2-2K>0$. 

There are many other constructions which produce torque springs. Another
configuration, which is closer in design to a normal spring, is that
given in \figref{fig:2}. In two-dimensions this type of construction may
be preferable to that in  \figref{fig:1} to reduce the number of spring
intersections when assembing a network of torque springs. It also may be
preferable if we wish to attach a torque spring say between two parallel
interfaces.

The torque springs described here are quite floppy. To give them some
structural integrity one would need to add a scaffolding of additional
springs, extending out of the plane if the torque springs are going to
be used in a three dimensional network.  Provided the spring constants
of these additional springs are sufficiently small, this can be done
with only a small perturbation to the response of the torque spring, as
shown in \cite{Vasquez:2011:CCS}.

In assembling a network of torque springs it may happen that an interior
spring or interior node of one torque spring intersects with an interior
spring or interior node or interior node of another torque spring. Since
we have the flexibility to move the interior nodes of each torque spring
we only need be concerned with the intersection of two springs, or
between the intersection of one spring and a node. In three dimensions
if a spring intersects with another spring or a node we can replace one
or both springs by an equivalent truss of springs to avoid this
situation. In two dimensions if a spring intersects with a node we can
again replace the spring by an equivalent truss to avoid this situation.
Then if two springs intersect in two dimensions they must either overlap
or cross: if they overlap we can replace each by an equivalent truss of
springs, while if they cross we can (within the framework of linear
elasticity) place a node at the intersection point and appropriately
choose the spring constants of the joining springs so that they respond
like two non-interacting springs -- see example 3.15 in Milton and
Seppecher \cite{Milton:2008:RRM}.

%%%%%%%%%%%%%%%%%%%%%%%%%%%%%%%%%%%%%%%%%%%%%%%%%%%%%%%%%%%%%%%%%%%%%%%%
\begin{figure}
\begin{center}
\includegraphics[width=0.5\textwidth]{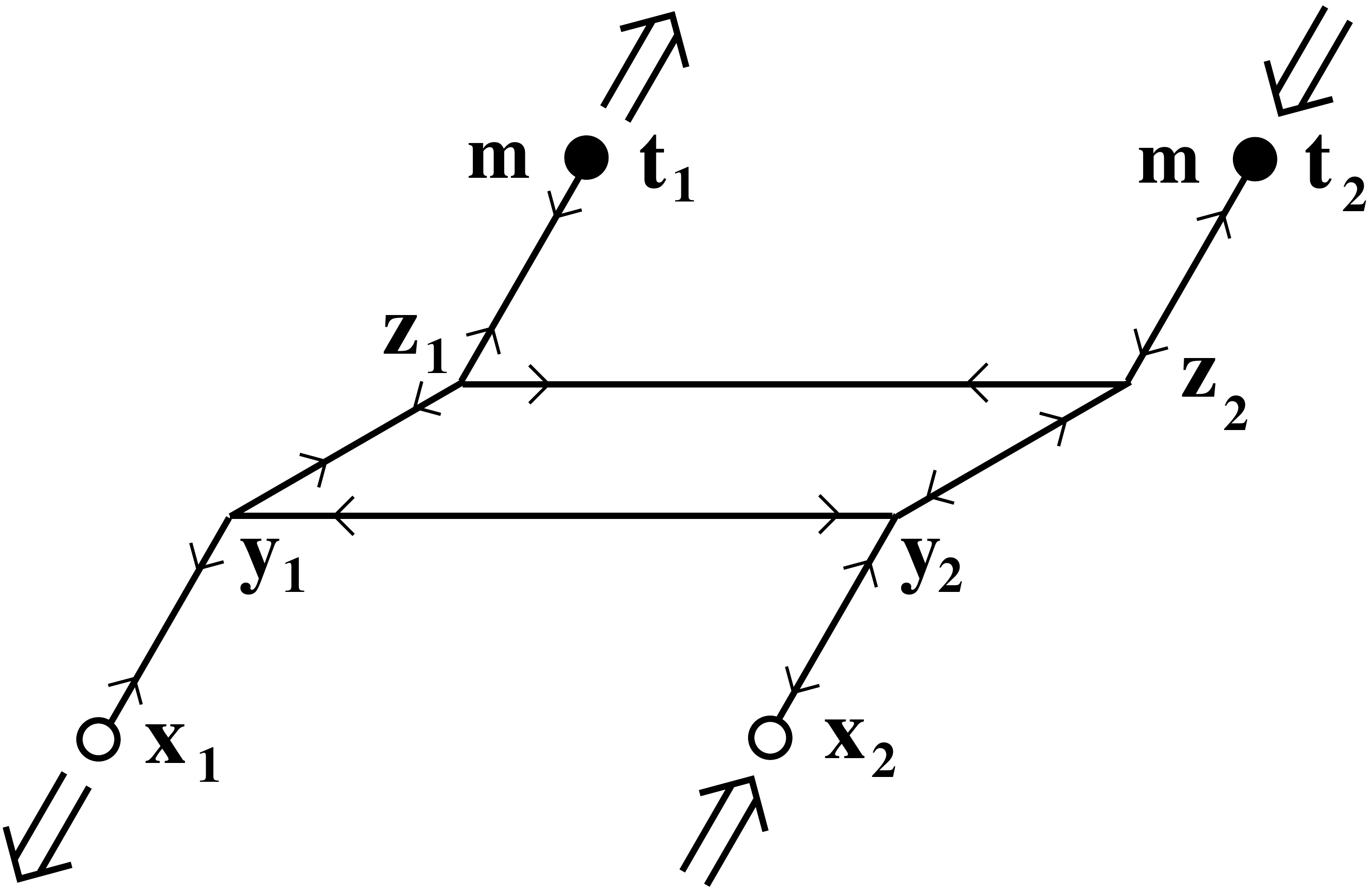}
\end{center}
\caption{Sketch of a torque spring. The open circles represent terminal
nodes, and the closed circles could be either terminal nodes or interior
nodes with masses attached. The straight lines represent springs. The
large arrows represent external or inertial forces acting on the nodes
at one instant in time. The two small arrows on each spring give the
direction of the force which the spring exerts on the node nearest to
the arrow.}
\label{fig:1}
\end{figure}

%%%%%%%%%%%%%%%%%%%%%%%%%%%%%%%%%%%%%%%%%%%%%%%%%%%%%%%%%%%%%%%%%%%%%%%%
\begin{figure}
\begin{center}
\includegraphics[width=0.5\textwidth]{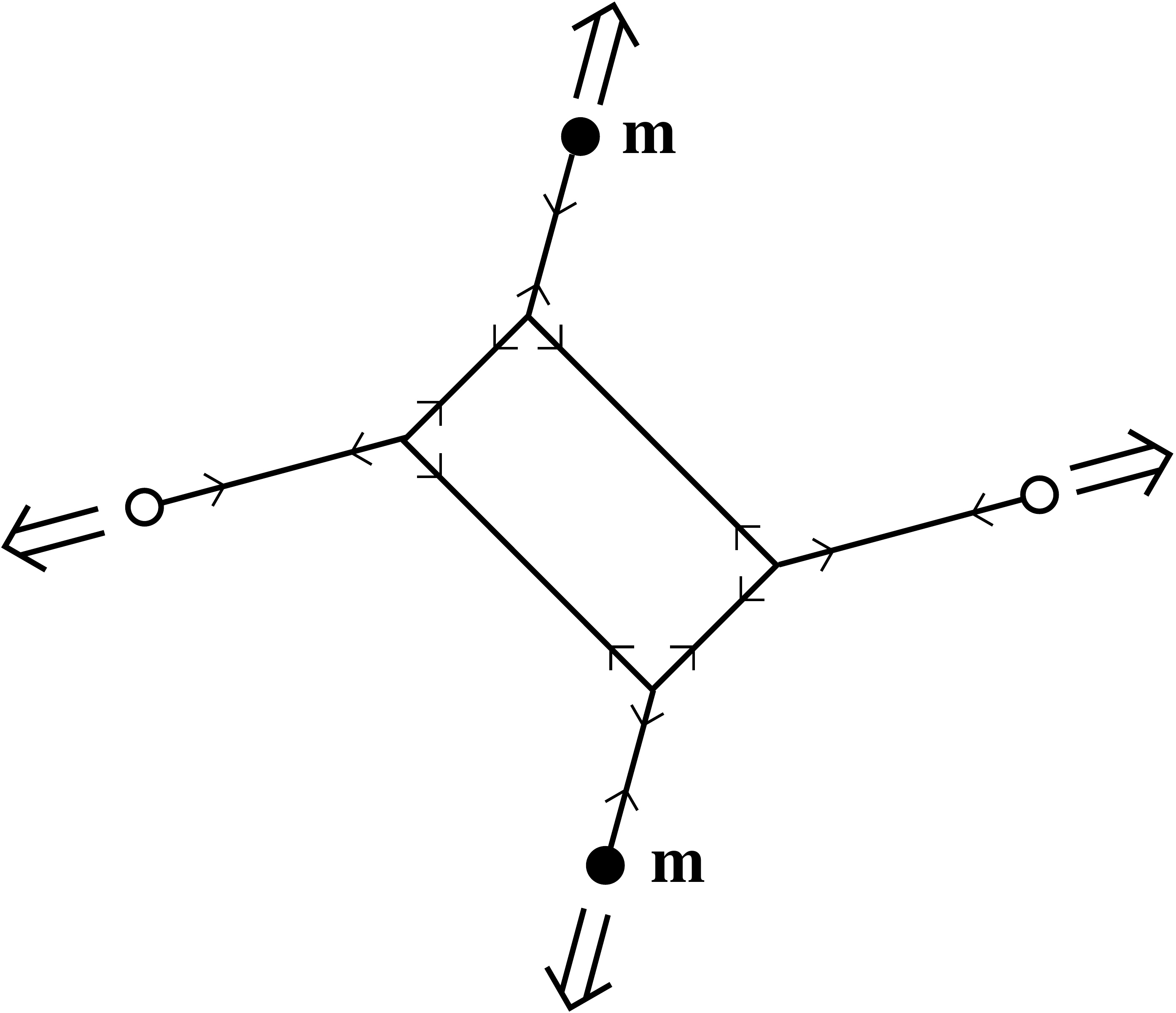}
\end{center}
\caption{An alternative construction of a torque spring. The straight
lines represent springs and the circles, large arrows, and small arrows
have the same meaning as in figure 1.}
\label{fig:2}
\end{figure}

%%%%%%%%%%%%%%%%%%%%%%%%%%%%%%%%%%%%%%%%%%%%%%%%%%%%%%%%%%%%%%%%%%%%%%%
\subsection{Homogenization of a discrete network of\\* torque springs}
As shown in \secref{sec:dte}, the original network of springs with nodes
at positions $\Bx_1, \Bx_2, \Bx_3, \ldots, \Bx_n \ldots$ and spring
constants $k_{ij}$ responds in an equivalent manner to the new network
of torque springs with nodes at positions $\Bx'_1, \Bx'_2, \Bx'_3,
\ldots, \Bx'_n \ldots$ and torque spring parameters given by \eqref{a.14}.
If the original network of springs homogenizes to an effective
elasticity tensor field $\BC(\Bx)$ then the new network of torque
springs homogenizes to an effective elasticity tensor field
$\BC'(\Bx)$ given by \eqref{a.7}, assuming the transformation $\Bx'(\Bx)$
only has variations on the macroscopic scale. In particular the stress
field in the homogenized network of torque springs is not be
symmetric, and is influenced not just by the local strain, but also
by the field of microrotations.

There are some practical barriers to this homogenization. Suppose, for
simplicity, that we are in two dimensions, that the original network
consists of a triangular network of identical springs with bond length
$h$ under uniform loading so that the tension is the same in all
springs, and that the transformation is a rigid rotation $\Bx'=\BR\Bx$
where $\BR^T\BR=\BI$. The displacement $\Bu_i$ of the nodes $\Bx_i$ is,
up to a translation, that of uniform dilation, $\Bu_i=\Ga\Bx_i$.  It
follows that if $i$ and $j$ are adjacent nodes on the network, then
$\Bu'_i-\Bu'_j$ scales in proportion to $h$. On the other hand, in order
that the traction force per unit length on a line remains constant the
tension $T$ in each torque spring must also scale in proportion to $h$.
Therefore the torque spring constant $k'=Km\Go^2/(m\Go^2-2K)$ must be
essentially independent of $h$. Also we don't want the density of mass
per unit area associated with the torque springs to be too large
(otherwise gravitational forces would be very significant). This would
be ensured if $m$ scales as $h^{\beta}$ where $\beta\geq 2$. Since 
\begin{equation} 
K=\frac{k'm\Go^2}{2k'+m\Go^2} 
\label{a.19}
\end{equation}
we see that $K$ should also scale as $h^{\beta}$, and that $2K$ would be
close $m\Go^2$ when $h$ is small. Thus each torque spring is very
close to resonance. If this is satisfied at one frequency, it will not
be satisfied at nearby frequencies. Thus the metamaterial is 
operational only within an extremely narrow band of frequencies. The
situation is similar in three dimensions in a network having bond
lengths of the order of $h$.  Then $\Bu'_i-\Bu'_j$, $T$, $k'$ and $m$
need to scale as $h$, $h^2$, $h$, and $h^{\beta}$, respectively, with
$\beta\geq 3$ to avoid an infinite mass density in the limit $h\to 0$.
($T$ must scale as $h^2$ to maintain a constant traction per unit area
on a surface).  Again $K$ given by \eqref{a.19} must be close to $m\Go^2/2$
when $h$ is small.

In three dimensions an alternative is to avoid the use of masses within
each torque spring altogether. This can be achieved by pinning the
internal nodes of the torque springs, where there would be masses (such
as at the nodes $\Bt_1$ and $\Bt_2$ in \figref{fig:1}), to a rigid
lattice (designed in a way which avoids intersection with the springs
inside the torque springs). Such a pinning corresponds to setting
$m=\infty$ and each torque spring has then a spring constant
$k'=K$ which is independent of frequency. The resulting metamaterial
is operational at all frequencies. Note that within the framework
of linear elasticity each torque spring exerts a torque but not a net
force on the underlying rigid lattice. If the rigid lattice (which might
have only finite extent) itself is not pinned we require that the
external forces on the metamaterial to be such that there is no net
overall torque on the rigid lattice.

A more serious concern is the validity of linear elasticity, at least
using the torque spring designs proposed here. A characteristic feature
of the designs involving masses is that the internal masses $m$ do not
move when the springs are translated, to first order in the
displacement.  This accounts for the balance of forces
$\BF'_{j,i}=-\BF'_{i,j}$. However the masses do move significantly if
the terminals are translated a distance which is comparable to the size
of the torque spring. Alternatively, if we pin the internal nodes of the
torque springs, where there would be masses, to a rigid lattice then
this restricts the motion of the torque spring terminals relative to the
lattice. Clearly for the operation of the metamaterial the displacements
$\Bu'_i$ must be small compared to $h$, assuming the size of each torque
spring is of order $h$.  When $h$ is very small this severely limits the
amplitude of waves propagating in the metamaterial for which linear
elasticity applies.  Thus the only metamaterials of the type described
here that might possibly be of practical interest are those  for which
$h$ is not too small.  This is in contrast to homogenization of a normal
elastodynamic network where linear elasticity may apply when only the
displacement differences $\Bu'_i-\Bu'_j$, between adjacent nodes  $i$
and $j$, are small compared to $h$.

%%%%%%%%%%%%%%%%%%%%%%%%%%%%%%%%%%%%%%%%%%%%%%%%%%%%%%%%%%%%%%%%%%%%%%%%%%%%%%%%%%
\section{Active exterior cloaking in the quasistatic regime}
\label{sec:static}
We show that for the Laplace equation, it is possible for a device to
generate fields that cancel out the incident field in a region while not
interfering with the incident field far away from the device.  Our
results generalize to the quasistatic (low frequency) regime. Thus any
(non-resonant) object located inside the region where the fields are
negligible interacts little with the fields and is for all practical
purposes invisible.  Here we relate the problem of designing a cloaking
device to the classic problem of approximating a function with
polynomials. Then we propose a cloak design that is based on a family of
polynomials. We also show that our solution can be easily modified to
give cloak objects while giving the illusion of another object (illusion
optics as in \cite{Lai:2009:IOO}).

%%%%%%%%%%%%%%%%%%%%%%%%%%%%%%%%%%%%%%%%%%%%%%%%%%%%%%%%%%%%%%%%%%%%%%%%%%%%%%%%%%
\subsection{Active exterior cloak design} 
Following the ideas presented in \cite{Vasquez:2009:AEC}, we first state the
requirements that the field generated by a device (source) needs to
satisfy in order to cloak objects inside a predetermined region. Here we
denote by $B(\Bx,r)\subset \real^2$ the open ball of radius $r>0$
centered at $\Bx\in\real^2$.

Let $B(\Bc,a)$ with $a>0$ and $\Bc \in \real^2$ be the region where we
want to hide objects (the cloaked region). The cloaking device is an
active source (antenna) located (for simplicity) inside $B(0,\delta)$ with
$\delta \ll 1$. Assuming a priori knowledge of the incident (probing)
potential $u_0$, we say that the device is an active exterior cloak for
the region $B(\Bc,a)$ if the device generates a potential $u$ such that
 \begin{enumerate}[i.]
  \item The total potential $u+u_0$ is very small in the cloaked region
 $B(\Bc,a)$.
  \item The device potential $u$ is very small outside $B(0,R)$, for
  some large $R>0$.
 \end{enumerate}

Therefore, if the incoming (probing) field is known in advance, an
active exterior cloak hides both itself and any (non-resonant)
object placed in the region $B(\Bc,a)$. Indeed any object inside
$B(\Bc,a)$ only interacts with very small fields and the device field is
very small far away from the device.

After a suitable rotation of axes, we may assume, without loss of
generality, that $\Bc=(p,0)$ with $p>0$.  As in \cite{Vasquez:2009:AEC},
we require the following conditions in our cloak design, 
\begin{equation}
\begin{aligned}
p>a+\delta, & ~~\text{the active device is outside the region $B(\Bc,a)$, and}\\
R>a+p, & ~~\text{the cloaking effect is observed in the far field.}
\end{aligned}
\label{eq:req}
\end{equation}

%%%%%%%%%%%%%%%%%%%%%%%%%%%%%%%%%%%%%%%%%%%%%%%%%%%%%%%%%%%%%%%%%%%%%%%%
\subsection{The conductivity equation}
Next, in the spirit of \cite{Vasquez:2009:AEC,Vasquez:2011:MAA} we give a more rigorous
formulation of the exterior cloaking problem for the two-dimensional
conductivity equation and prove its feasibility.  The results extend
easily to the quasistatic regime.

\begin{figure}
 \begin{center}
  \includegraphics[width=\textwidth]{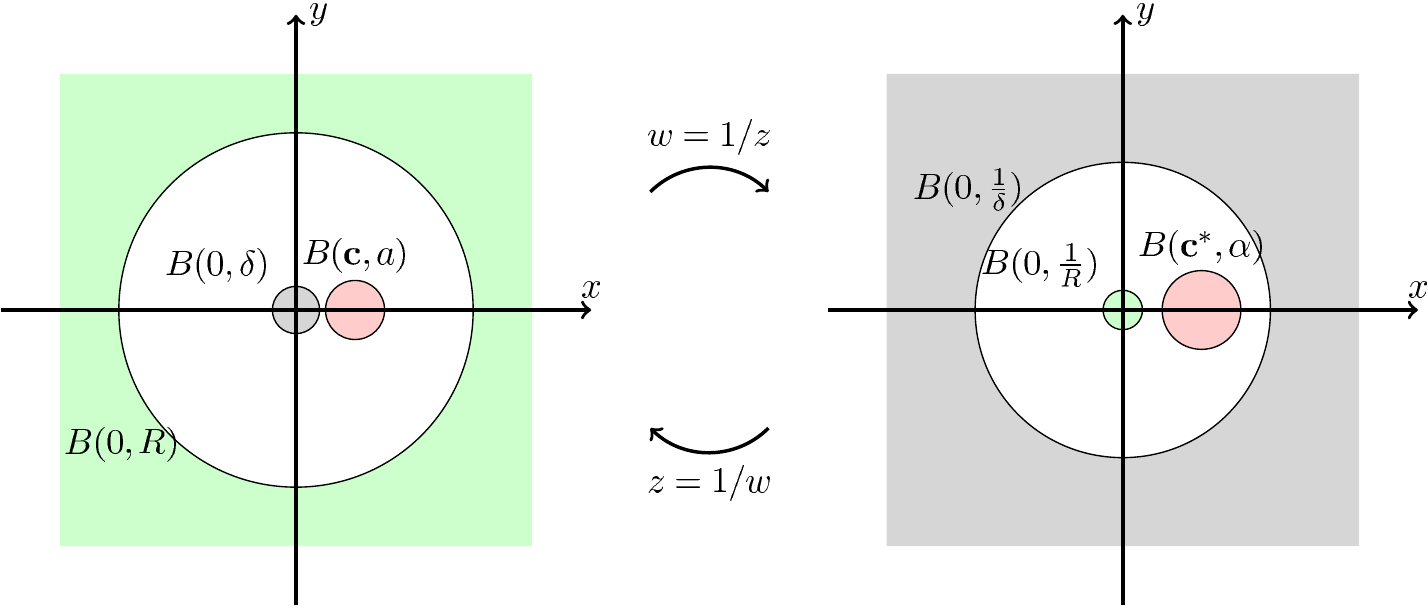}
 \end{center}
 \caption{The regions of Theorem~\ref{thm-1} (left) and their transforms
 under the Kelvin or inversion transformation (right). The device field
 is harmonic everywhere except in the gray areas, is close to minus the
 incident field in the red region and is close to zero in the green
 region.}
 \label{fig:kelvin}
\end{figure}

\begin{theorem}
 \label{thm-1}
 Let $a$, $\Bc$, $R$ and $\delta$ satisfy \eqref{eq:req}, then for any
 $\epsilon>0$ and any harmonic
 potential $u_0$, there exists a function $g_0:\real^2 \to \real$ and a
 potential $u:\real^2 \to \real$, satisfying
 \begin{equation}
  \left\{
  \begin{aligned}
   \varDelta u  &=  0, ~\text{in $\real^2\setminus \overline{
   B(0,\delta)}$,}\\
   u &= g_0,  ~\text{on $\partial B(0,\delta)$,}\\
   |u| &< \epsilon ~\text{in $\real^2\setminus B(0,R)$,}\\
   |u+u_0|&<\epsilon ~\text{in $\overline{B(\Bc,a)}$.} 
  \end{aligned}
  \right.
  \label{eq:g0}
 \end{equation}
\end{theorem}

\begin{proof}
By applying the inversion (or Kelvin) transformation $\displaystyle w
\doteq 1/z$, the geometry of problem \eqref{eq:g0} transforms as
follows, 
\begin{itemize}
 \item $\real^2 \setminus B(0,\delta)$ transforms to
 $B(0,1/\delta)$,
 \item $\real^2 \setminus B(0,R)$ transforms to
 $B(0,1/R)$,
 \item $B(\Bc,a)$ transforms to $B(\Bc^*,\alpha)$, with
 \[
  \alpha= \frac{a}{|p^2-a^2|},~~ 
  \Bc^{*}=(\beta,0), ~~\text{and}~~ 
  \beta=\frac{p}{p^2-a^2}.
 \]
\end{itemize}
The different regions and their transforms are illustrated in
\figref{fig:kelvin}.
Thus the problem \eqref{eq:g0} is equivalent to finding $\tg_0$ and $\tu$
such that
 \begin{equation}
  \left\{
  \begin{aligned}
  \varDelta \tu &= 0, ~\text{in $B(0,1/\delta)$},\\
  \tu &= \tg_0, ~\text{on $\partial B(0,1/\delta)$},\\
  |\tu| &< \epsilon, ~\text{in $\overline{B(0,{1/R})}$},\\
  |\tu+\tu_0| &< \epsilon, ~\text{in $\overline{B(\Bc^*,\alpha)}$.}
  \end{aligned}
  \right.
  \label{eq:g0inv}
 \end{equation}
 Relating to the functions $g_0$ and $u_0$ from \eqref{eq:g0}, we get
 $\tg_0(z)=g_0(1/z)$ and $\tu_0(z)=u_0(1/z)$, so that $\tu_0$ is
 harmonic in the whole space except the origin.  Next, we observe that
 the inversion transforms the necessary conditions 
 \eqref{eq:req} to
 \begin{equation}
  \begin{aligned}
  \frac{1}{R}&<\beta-\alpha, ~ \text{the two balls $B(0,1/R)$ and
  $B(\Bc^*,\alpha)$ do not touch},\\
 \beta+\alpha&<\frac{1}{\delta},~\text{the two balls $B(0,1/\delta)$ and
 $B(\Bc^*,\alpha)$ do not touch}.
  \end{aligned}
  \label{eq:reqinv}
 \end{equation} 
Let $\tU_0$ be the analytic extension of $\tu_0$ in $B(\Bc^*,\alpha)$,
obtained with the harmonic conjugate such that $\tu_0$ is the real part
of $\tU_0$. Because of analyticity of $\tU_0$, we can approximate
$\tU_0$ with a polynomial $Q_0$ (e.g. by truncating the series expansion
of $\tU_0$) such that
\begin{equation}
 |\tU_0-Q_0| < \frac{\epsilon}{2} 
 ~\text{in $\overline{B(\Bc^*,\alpha)}$}.
\label{3'''} 
\end{equation}
This immediately yields the approximation for $\tu_0$
\begin{equation}
 |\tu_0-q_0| < \frac{\epsilon}{2} 
 ~\text{in $\overline{B(\Bc^*,\alpha)}$},
 \label{3''} 
\end{equation}
where $q_0\doteq \Re(Q_0)$, i.e., the real part of $Q_0$. Since $\tU_0$ can be
approximated arbitrarily well by a polynomial, it is enough to consider
\eqref{eq:g0inv} when $\tu_0$ is the real part of a polynomial, i.e.
 \begin{equation}
  \left\{
  \begin{aligned}
  \varDelta \tu &= 0, ~\text{in $B(0,1/\delta)$},\\
  \tu &= \tg_0, ~\text{on $\partial B(0,1/\delta)$},\\
  |\tu|&<\epsilon, ~\text{in $\overline{B(0,1/R)}$},\\
  |\tu+q_0|&< \epsilon/2, ~\text{in $\overline{B(\Bc^*,\alpha)})$}.
  \end{aligned}
  \right.
  \label{eq:g0invpoly} %4
 \end{equation}
In other words, problem \eqref{eq:g0invpoly} is equivalent to finding a
function $\tu$, harmonic inside $B(0,1/\delta)$ that approximates $q_0$
well inside $B(\Bc^*,\alpha)$ but is practically zero in $B(0,1/R)$.

Let us now recall a classic result in harmonic approximation
theory due to Walsh (see \cite{gardiner:1995:HA}, page 8).

\begin{lemma}[Walsh]
\label{walsh} 
Let $K$ be a compact set in $\real^2$ such that $\real^2\setminus K$ is
connected. Then for each function $w$, harmonic on an open set
containing $K$, and for each $d>0$, there is a harmonic
polynomial $q$ such that $|w-q|<d$ on $K$.
\end{lemma}

Walsh's lemma implies the existence of
a harmonic solution to problem \eqref{eq:g0invpoly}. Indeed, from the
design requirements \eqref{eq:reqinv} there exists $0 < \xi \ll 1$ such that 
\begin{equation}
  \frac{1}{R}+\xi<\beta-\alpha-\xi.
 \label{eq:6}
\end{equation}
Then applying Lemma~\ref{walsh} with $K=\overline{B(0,1/R)} \cup
\overline{B(\Bc^*,\alpha)}$, we obtain that for an arbitrary small
parameter $0<d\ll 1$ and for the function $w$ satisfying 
\begin{equation}
w=\begin{cases}
   0 & \text{ in $B(0,\frac{1}{R}+\xi)$},\\
  -q_0 & \text{ in $B(\Bc^*,\alpha+\xi)$},
  \end{cases}
\end{equation}
there exists a harmonic polynomial $q$ such that $|q-w|<d$ on $K$.  We
conclude that there exists a harmonic solution to problem
\eqref{eq:g0invpoly}, which implies the statement of Theorem~\ref{thm-1}.
\qed
\end{proof}

%%%%%%%%%%%%%%%%%%%%%%%%%%%%%%%%%%%%%%%%%%%%%%%%%%%%%%%%%%%%%%%%%%%%%%%%%%%%%%%%%%
\subsection{Explicit polynomial solution in the\\* zero frequency regime}
Although mathematically rigorous, the existence result of Theorem \ref{thm-1}
(which follows from Walsh's lemma) does not give an explicit expression for the
required potential at the active device (antenna). In \cite{Vasquez:2009:AEC}
(see also \cite{Vasquez:2011:MAA}) we give a polynomial solution to problem
\eqref{eq:g0inv}. Unfortunately the radius $a$ of the cloaked region in the
polynomial solution of \cite{Vasquez:2009:AEC,Vasquez:2011:MAA} is limited by
the distance from the origin $p$ according to $a < (2+2\sqrt{2})^{-1}p$. Thus
in \cite{Vasquez:2009:AEC,Vasquez:2011:MAA} we can only cloak large objects if
they are sufficiently far from the origin. Here we state a conjecture that
extends our previous results \cite{Vasquez:2009:AEC,Vasquez:2011:MAA} and that
gives more freedom on the choice of the cloaked region location and size. This
is  supported by numerical evidence (see Figs.~\ref{fig:poly} and
\ref{fig:utot}).

\begin{conjecture}
\label{conj}
 Let $\Bc^* = (\beta,0)$ be as in the proof of Theorem~\ref{thm-1}. For
 any $L>0$, any disk $S_1$ in the connected component containing the
 origin of the set
 \begin{equation}
 D_{\beta,L}\doteq\Mcb{
  z\in \complex,~
       |z-\beta|^L |z|<\frac{\beta^{L+1}L^L}{(L+1)^{L+1}}
},
\label{43} 
\end{equation}
any disk $S_2$ in the connected component of the set $D_{\beta,L}$
containing the point $\Bc^*$ and any $\epsilon>0$, there exists two
positive integers $s$ and $n$ such that $ | s/n - L| < \epsilon$ and the
polynomial $P_{n,s}:\complex \to \complex$ defined by 
\begin{equation}
  P_{n,s}(z)=\M{1-\frac{z}{\beta}}^s~\sum_{j=0}^{n-1}\M{\frac{z}{\beta}}^j{s+j-1\choose
j},
\label{44}
\end{equation} 
satisfies 
\begin{equation} 
 |P_{n,s}  - 1| < \epsilon  ~\text{on $\partial S_1$}
 ~~\text{and}~~
 |P_{n,s}| < \epsilon ~\text{ on $\partial
 S_2$.}
 \label{eq:rns}
\end{equation}
Moreover the approximation property \eqref{eq:rns} is not satisfied when
either $S_1$ or $S_2$ is not contained in $D_{\beta,L}$.
\end{conjecture}

\begin{remark}
To see why we expect that the polynomial $P_{n,s}$ satisfies
\eqref{eq:rns}, notice that $z=\beta$ is a root of multiplicity $s$ of
the polynomial $P_{n,s}$. From the Taylor expansion of $P_{n,s}$ around
$z=\beta$, we can expect that $P_{n,s} \approx 0$ in a sufficiently
small disk around $\beta$. (The symbol $\approx$ denotes an
approximation with respect to the supremum norm.)  Now the function
\[
 g(z) = \Mb{\sum_{j=0}^{n-1} \M{\frac{z}{\beta}}^j
 {s+j-1\choose j}} - \M{ 1 - \frac{z}{\beta}}^{-s},
\]
has a root of multiplicity $n$ at $z=0$, i.e.  $g^{(k)}(0) = 0$ for
$k=0,\ldots,n-1$. This is because the sum in the definition of $g(z)$
corresponds to the first $n$ terms in the Taylor expansion around $z=0$
of $(1-z/\beta)^{-s}$. Thus by Leibniz rule, $z=0$ is a root of
multiplicity $n$ of the polynomial $P_{n,s} - 1 = (1-z/\beta)^s g(z)$,
and we can expect $P_{n,s} \approx 1$ in a sufficiently small disk
around the origin. This suggests an alternative definition of $P_{n,s}$
as the unique Hermite interpolation polynomial (see e.g.
\cite{Stoer:2002:INA}) satisfying:
\[
 \begin{aligned}
 P_{n,s}(0) &= 1,\\
 P_{n,s}^{(k)}(0) & = 0, ~~\text{for $k=1,\ldots,n-1$},\\
 P_{n,s}^{(k)}(\beta) &=0, ~~\text{for $k=0,\ldots,s-1$}.
 \end{aligned}
\]
\end{remark}

\begin{remark}
To motivate our belief that the region $D_{\beta,L}$ is the region of
convergence of $P_{n,s}$ as $n\to \infty$ and $s\to \infty$ with $s/n \to
L$, consider the special case where $L>0$ is an integer and $s=nL$. Then
the last term in the sum \eqref{44} defining $P_{n,nL}(z)$ is
\[
 \M{1-\frac{z}{\beta}}^{nL} \M{\frac{z}{\beta}}^{n-1} { n(L+1) -1
 \choose n-1},
\]
which diverges outside of $D_{\beta,L}$ as $n\to \infty$ because 
\[
 { n(L+1) -1 \choose n-1}^{\frac{1}{n-1}} \to \frac{(L+1)^{L+1}}{L^L},
 ~ \text{as $n\to \infty$},
\]
which follows from Stirling's formula, see e.g. \cite[\S
5.11]{olver:2010:NHM}.
Therefore the region of divergence of $P_{n,Ln}$ contains the complement
of $D_{\beta,L}$.
\end{remark}

For some polynomial $Q_0$ we can deduce from  \eqref{eq:rns} that
\begin{equation} 
 Q_0P_{n,s}-Q_0 \approx 0 ~\text{on $\partial S_1$}
 ~~\text{and}~~
 Q_0P_{n,s}-Q_0\approx -Q_0 ~\text{ on $\partial
 S_2$,}
 \label{45}
\end{equation}
Thus, by construction, the real part of $W=Q_0P_{n,s}-Q_0$ is a solution
of \eqref{eq:g0invpoly} with $S_1 = B(0,1/R)$ and $S_2=
B(\Bc^*,\alpha)$.  In the particular case when $n=s$ (i.e.  $L=1$),
Conjecture~\ref{conj} was proved in
\cite{Vasquez:2011:MAA}, see also \cite{Vasquez:2009:AEC}.

%%%%%%%%%%%%%%%%%%%%%%%%%%%%%%%%%%%%%%%%%%%%%%%%%%%%%%%%%%%%%%%%%%%%%%%%
\begin{figure}
\begin{center}
\begin{tabular}{cc}
\includegraphics[width=0.45\textwidth]{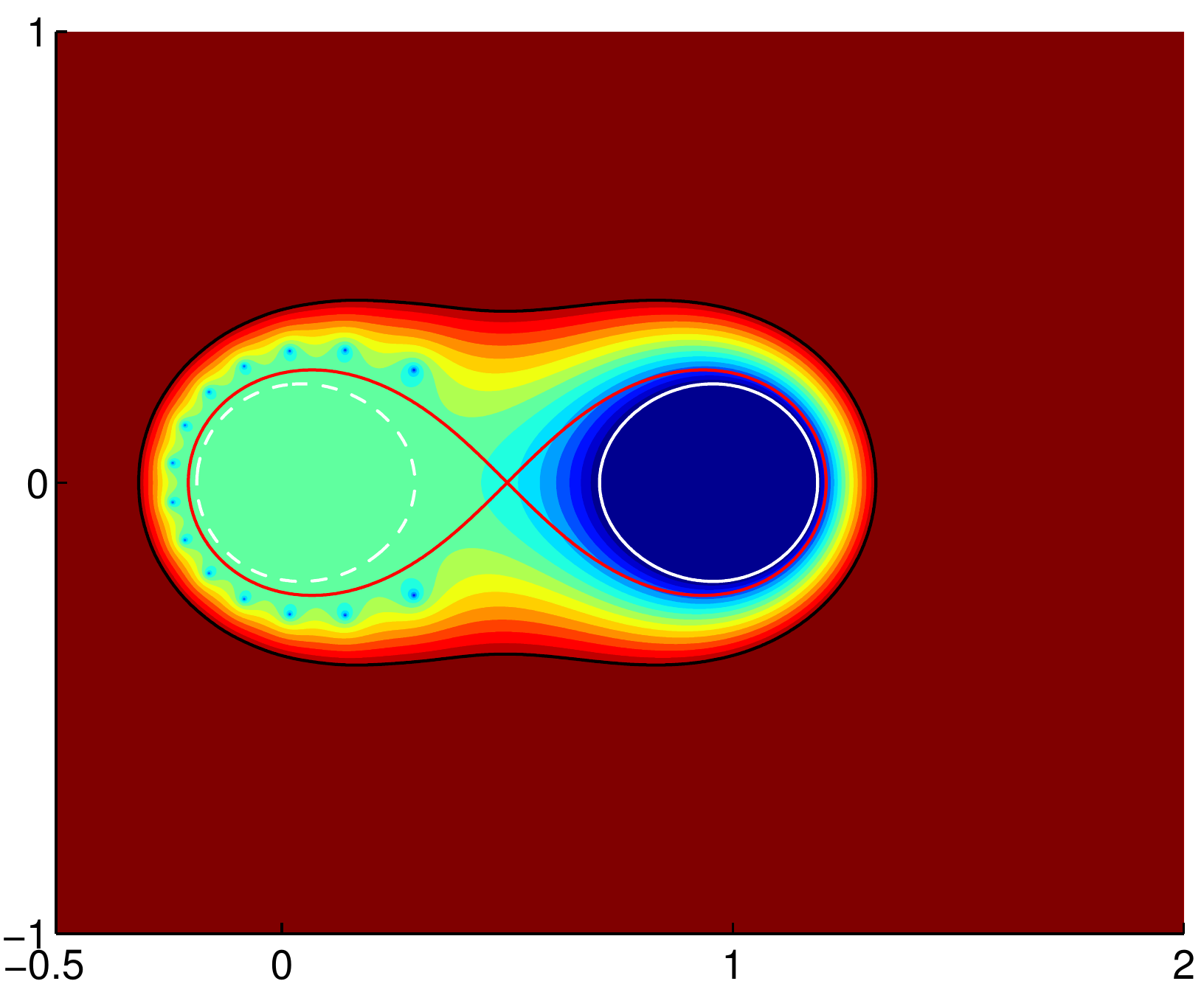} &
\includegraphics[width=0.45\textwidth]{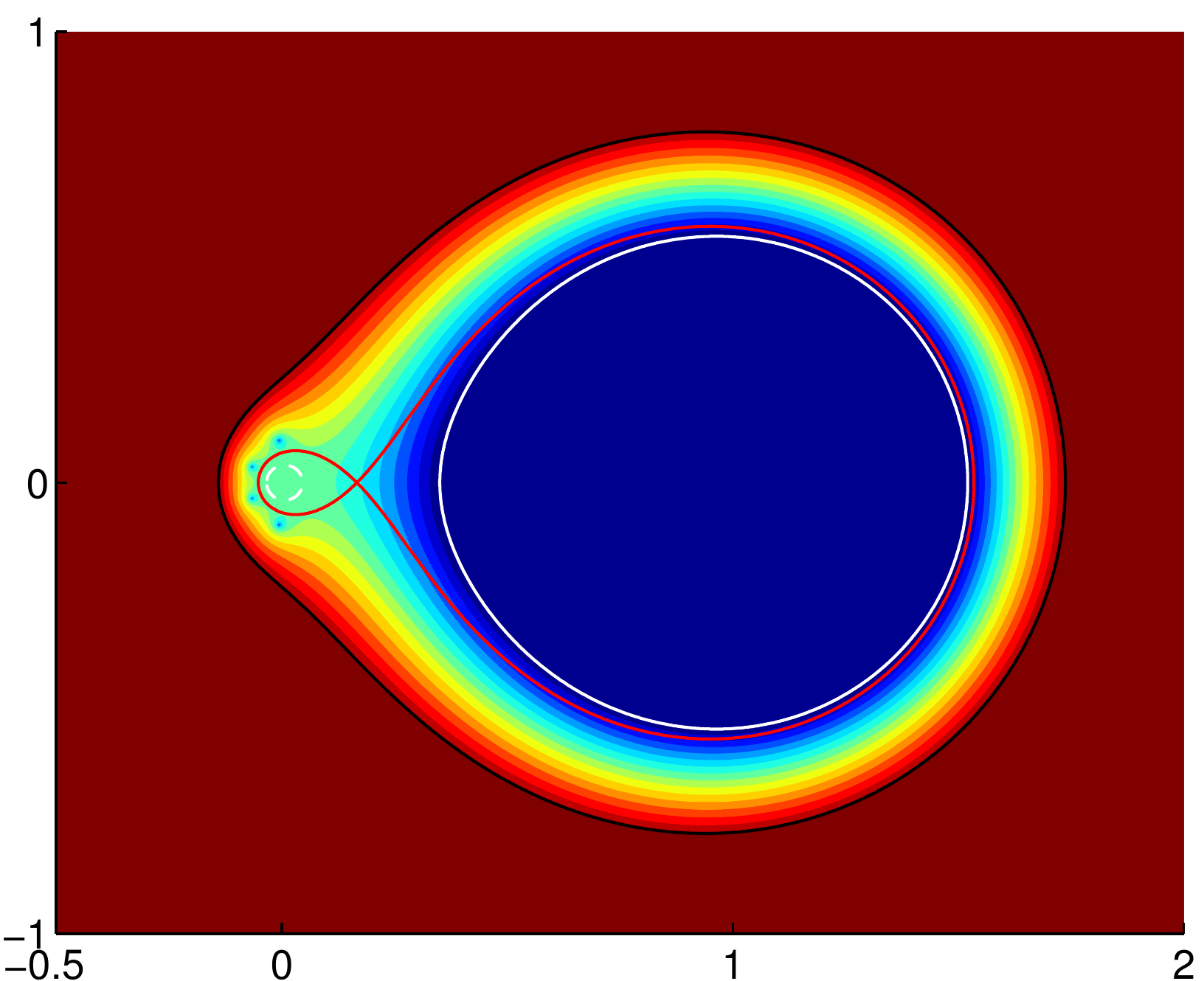}\\
(a) $n=s=15$ & (b) $n=5$, $s=25$
\end{tabular}
\end{center}
\caption{Contour plot for the polynomial $P_{n,s}$ with $\beta=1$. The solid
white line is the level-set $|P_{n,s}(z)|=10^{-2}$, thus the cloaked
region could be any disk inside this level-set. The dashed white line is
the level-set $|P_{n,s}(z)-1|=10^{-2}$. The device field is small in any
circle inside this level-set. The red curve is the boundary of
$D_{\beta,L}$, the conjectured region of convergence of $P_{n,s}$ as both $n\to
\infty$ and $s\to \infty$ with $s/n=1$ and $5$, respectively. (This is proved in
the case $n=s$ in \cite{Vasquez:2011:MAA}.) The color scale is
logarithmic from 0.01 (dark blue) to 100 (dark red), with light green
representing 1.} 
\label{fig:poly}
\end{figure}

%%%%%%%%%%%%%%%%%%%%%%%%%%%%%%%%%%%%%%%%%%%%%%%%%%%%%%%%%%%%%%%%%%%%%%%%
\begin{figure}
\begin{center}
\begin{tabular}{cc}
\includegraphics[width=0.45\textwidth]{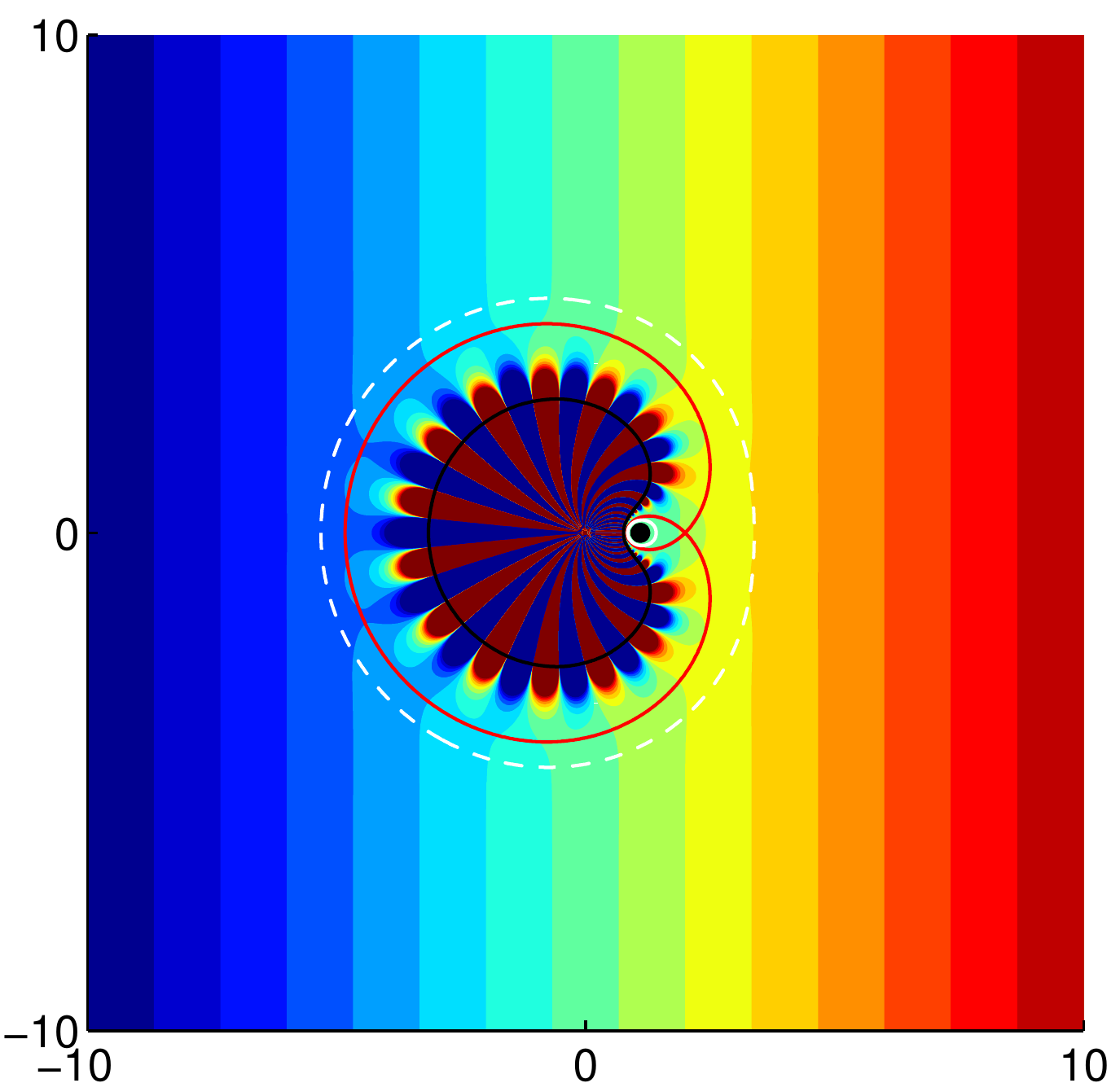} &
\includegraphics[width=0.45\textwidth]{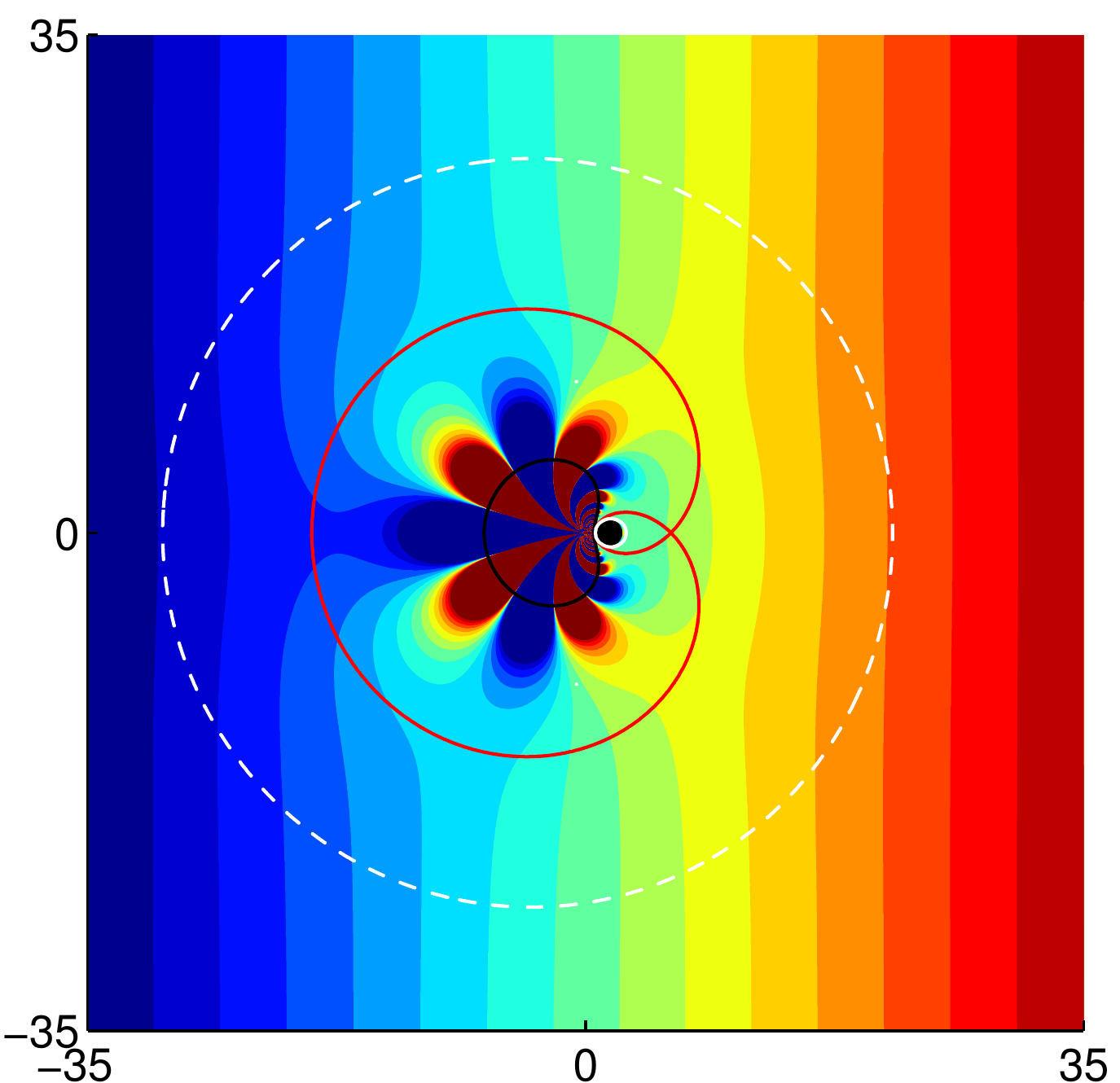}\\
(a) $n=s=15$ & (b) $n=5$, $s=25$
\end{tabular}
\end{center}
\caption{Real part of the total field with the cloaking device active,
incident field $u_0(x,y)=x$ and $\beta=1$. The solid white, dashed white
and red lines are the inversion (Kelvin) transforms of their counterparts in
\figref{fig:poly}. The small black scatterer inside the cloaked region
is an almost resonant disk centered at $(p,0)$ with radius $r$ and
dielectric constant $\epsilon$ given by: (a) $p=1.1$, $r=0.2$ and
$\epsilon=-0.99$; (b) $p=1.7$, $r=0.9$ and $\epsilon=-0.998$. In (a) the color
scale is linear from -10 (dark blue) to 10 (dark red). In (b)
the scale is linear from -35 to 35.}
\label{fig:utot}
\end{figure}

In \figref{fig:poly} we present a contour plot of the polynomial
$P_{n,s}$ when $\beta=1$ for different values of $n$ and $s$. The region
bounded by the peanut shaped red curve represents the conjectured domain
of convergence $D_{\beta,L}$ of the functions $P_{n,s}$ when
$n\rightarrow \infty$, $s\rightarrow \infty$ and $s/n \rightarrow L$.
In the left side of $D_{\beta,L}$, $P_{n,s}$ is conjectured to
converge to one, while in the right side of $D_{\beta,L}$, $P_{n,s}$ is
conjectured to converge to zero. The area within the solid white circle
on the right represents the region to be cloaked and the area within the
dashed white circle in the left represents the location of the observer.
We now present an active cloak design based on Conjecture~\ref{conj}.
\begin{remark} 
Let $u_0$ be an a priori determined incoming harmonic potential.  Let
$n$ and $s$ be such that the real part of $W=Q_0P_{n,s}-Q_0$ is a
solution of \eqref{eq:g0invpoly} (recall $Q_0$ is a polynomial
approximation of $\tU_0$ in $B(\Bc^*,\alpha)$). Let $S$ be
a bounded region in the complex plane compactly including the two disks
$S_1=B(0,{\frac{1}{R}})$ and $S_2=B(\Bc^*,\alpha)$. Then,  the cloaking
strategy we propose consists of an active device (antenna) located
inside $B(0,\delta)$ and capable of generating  a potential equal to the
real part of $W({1\over z})$ on the set $\{z\in \complex, {1\over z}\in
\partial S\}$. By \eqref{3''}, the total potential in the original
physical configuration (the field from the antenna plus $u_0$) is well
approximated by the real part of $(W+Q_0)(1/z)$ which ensures an almost zero
field region in $B(\Bc,a)$ with negligible perturbations on the field
outside $B(0,R)$.
\end{remark}

\figref{fig:utot} illustrates how the cloaking device (represented by
the solid black curve) works after applying the back-inversion to the
configurations presented in \figref{fig:poly}. Here the incident field
is $u_0(x,y)=x$ and the objects we want to hide are almost resonant disks.
Clearly, the active device generates the necessary field to cancel the
field in the cloaked region while having a very small effect
in the far field (outside the white dashed circle). With a polynomial of
the same degree, when $s=5n$ (\figref{fig:utot}(b)) we can hide an
object roughly four times larger than when $s=n$ (\figref{fig:utot}(a)).
Thus using the polynomials $P_{n,s}$ allows us to cloak large objects
without restrictions on the distance from the origin $\beta$ as was the
case in \cite{Vasquez:2009:AEC,Vasquez:2011:MAA}. The disadvantage is
that cloaking is enforced on $\partial B(0,R)$ (dotted white line in
\figref{fig:utot}) with a larger $R$ in the asymmetric $L>1$ case than
in the symmetric $L=1$ case. For example to get a device field such that
$|u|<10^{-2}$, $R$ needs to be roughly five times larger when $s=5n$
(\figref{fig:utot}(b)) than when $s=n$ (\figref{fig:utot}(a)).

%%%%%%%%%%%%%%%%%%%%%%%%%%%%%%%%%%%%%%%%%%%%%%%%%%%%%%%%%%%%%%%%%%%%%%%%
\subsection{Extensions and applications}
We now extend the previous results to the case of an incoming field
having sources in $\real^2 \setminus \overline{B(0,R)}$.

\begin{remark}
\label{rem-1}The case studied in Theorem \ref{thm-1} (with an explicit
solution in Conjecture~\ref{conj}) corresponds to an incoming
field $u_0$ generated by a source located at infinity. The more
general case corresponding to an incoming field having sources in $\real^2\setminus
\overline{B(0,R)}$ can be treated similarly. Indeed,  the
problem remains to find $g_0$ and $u$ satisfying \eqref{eq:g0}, or
equivalently $\tg_0$ and $\tu$ satisfying \eqref{eq:g0inv} where inside
$B(0,{1/R})$, $\tu_0(z) = u_0(1/z)$ is harmonic. We can still
approximate its analytic extension $\tU_0$ by a polynomial in
$B(\Bc^*,\alpha)$ and the proof goes as in Theorem~\ref{thm-1}.
\end{remark}

Although our main focus here is cloaking, the same ideas can be applied
to illusion optics, where one wants to conceal an object by imitating
the response (scattering) of a completely different object.

\begin{remark}
\label{rem-2}
Let $u_1$ be the response of an object we wish to imitate, i.e. an
arbitrary potential harmonic in a set $D_1\subset
\real^2$ such that $\real^2\setminus B(0,R) \Subset D_1$. Assuming
the same notations as before, for any (known a priori) probing field $u_0$,
harmonic in $\real^2$, there exists a function $g\in C(\partial
B(0,\delta))$ so that the field $u$ generated by the active device (antenna) located
in $B(0,\delta)$ satisfies:
\begin{enumerate}[i.]
 \item The total field $u+u_0$ is very small in the cloaked region
 $B(\Bc,a)$.
 \item The device field $u$ is close to $u_1$ in $\real^2\setminus
 B(0,R)$.
\end{enumerate}
\end{remark}
Remark~\ref{rem-2} follows from the inversion (Kelvin) transform and
Lemma~\ref{walsh} by using an argument similar to the proof of
Theorem~\ref{thm-1}.  Using ideas similar to those in
Remark~\ref{rem-1}, the result of Remark~\ref{rem-2} can be generalized
to the case of an incoming field with sources in $\real^2\setminus
\overline{B(0,R)}$.

To illustrate Remark~\ref{rem-2} assume that the field $u_1$ is chosen
to be the response field of an inhomogeneity $\CI$ when probed with the
incident field $u_0$. Then Remark~\ref{rem-2} means that when probing
with the field $u_0$, an observer located in the far field detects the
inhomogeneity $\CI$ regardless of the inclusion inside $B(\Bc,a)$ and
without detecting the active illusion device. This creates the illusion that
the object inside $B(\Bc,a)$ is the inhomogeneity $\CI$.

%%%%%%%%%%%%%%%%%%%%%%%%%%%%%%%%%%%%%%%%%%%%%%%%%%%%%%%%%%%%%%%%%%%%%%%%
\section{Active exterior cloaking for the Helmholtz equation in\\* three
dimensions} 
\label{sec:helm} 
Previously in \cite{Vasquez:2009:AEC,Vasquez:2009:BEC} we designed
cloaking devices generating fields close to minus the incident field in
the region to be cloaked and vanishing far away from the devices.
Miller \cite{Miller:2007:PC} proposed an active cloak based on Green's
identities: a  single and double layer potential is applied to the
boundary of the cloaked region to cancel out the incident field inside
the cloaked region, while not radiating waves. The idea of using Green's
identities to cancel out waves in a region is well known in acoustics
(see e.g.
\cite{Ffowcs:1984:RLA,Malyuzhinets:1964:TAF,Jessel:1972:ASA}). Jessel
and Mangiante \cite{Jessel:1972:ASA} showed that it is possible to
achieve a similar effect to Green's identities (and thus cloaking) by
replacing the single and double layer potentials on a surface by a
source distribution in a neighborhood of the surface. What makes our
approach different is that the cloaking devices are multipolar sources
{\em exterior} to the cloaked region and thus do not completely enclose
the cloaked region. In \cite{Vasquez:2009:AEC,Vasquez:2009:BEC} the
cloaking devices are determined by solving numerically a least-squares
problem with linear constraints.  Our cloaking approach easily
generalizes to several frequencies \cite{Vasquez:2009:BEC} but requires
a priori knowledge of the incident field. Zheng, Xiao, Lai and Chan
\cite{Zheng:2010:EOC} used the same principle to achieve illusion optics
\cite{Lai:2009:IOO} with active devices, i.e. making an object appear as
another one.  Then in \cite{Vasquez:2011:ECA} we showed Green's identity
can be used to design devices which can cloak or give the illusion of
another object, i.e. achieving an effect similar to the active devices in
\cite{Vasquez:2009:AEC,Vasquez:2009:BEC,Zheng:2010:EOC}. The single and
double layer potential needed to reproduce a smooth field inside a
region while being zero outside is given by Green's identity and can be
replaced by a few multipolar sources using addition formulas for
spherical outgoing waves.  If in addition we want to imitate the
scattered field from an object as in \cite{Zheng:2010:EOC}, a similar
procedure applies.

The active cloaking devices we designed in
\cite{Vasquez:2009:AEC,Vasquez:2009:BEC,Vasquez:2011:ECA} are two
dimensional. Here we extend the result in \cite{Vasquez:2011:ECA} to the
Helmholtz equation in three dimensions. The wave pressure field
$u(\Bx)$ solves the Helmholtz equation,
\[
 \varDelta u + k^2 u = 0, ~~ \text{for $\Bx\in\real^3$},
 \label{eq:helm}
\]
where $k = 2\pi/\lambda$ is the wavenumber, $\lambda = 2\pi c / \omega$
is the wavelength, $c$ is the wave propagation speed (assumed to be
constant) and $\omega$ is the angular frequency. Recall for
future reference that the radiating Green's function for the Helmholtz
equation in three dimensions is 
\begin{equation}
 G(\Bx,\By) =  \frac{\exp[ik\abs{\Bx-\By}]}{4\pi\abs{\Bx-\By}}
 \label{eq:gf}
\end{equation}
Another underlying assumption is that the frequency $\omega$ is not a
resonant frequency of the scatterer we wish to hide.

%%%%%%%%%%%%%%%%%%%%%%%%%%%%%%%%%%%%%%%%%%%%%%%%%%%%%%%%%%%%%%%%%%%%%%%%
\subsection{Green's formula cloak}
As pointed out by Miller \cite{Miller:2007:PC} it is possible to cloak an object inside a
bounded region $D \in \real^3$ from an incident wave (probing field)
$u_i$ by generating a cloaking device field using monopole and dipole
sources (single and double layer potential) on $\partial D$. The device
field $u_d$ can be defined using Green's formula
\begin{equation}
\begin{aligned}
u_d(\Bx) & = \int_{\partial D} \D S_\By \Mcb { - (\Bn(\By) \cdot \nabla_\By
u_i(\By))
G(\Bx,\By) + u_i(\By) \Bn(\By) \cdot \nabla_\By G(\Bx,\By) }\\
 &= 
 \begin{cases}
  -u_i(\Bx), & \text{if $\Bx \in D$}\\
  0, & \text{otherwise},
 \end{cases}
\end{aligned}
\label{eq:green}
\end{equation}
so that the total field $u_i + u_d$ is a solution to Helmholtz equation
for $\Bx \notin \partial D$  that vanishes inside $D$ while being
indistinguishable form $u_i$ outside $D$. Since the waves reaching a
scatterer inside the cloaked region $D$ are practically zero, the
resulting scattered field is also practically zero. For 
clarity we assume the region $D$ is a polyhedron. The arguments
we give here can be easily modified for other domains with Lipschitz
boundary, as Green's identity \eqref{eq:green} is valid for these
domains \cite{evans:1992:MTF}.

\begin{remark}
 The Green representation formula \eqref{eq:green} requires that $u_i$
 be a $C^2$ solution to the Helmholtz equation inside $D$. A similar
 identity holds when $u_i$ is a $C^2$ radiating solution to the
 Helmholtz equation {\em outside} $D$. In this case, the device field
 $u_d$ vanishes inside $D$ and is identical to $-u_i$ outside $D$. The
 exterior cloak we present here can in principle be used to conceal
 a known active source and possibly accompanying scatterers inside $D$.
 If the radiating wave $u_i$ is taken to be the scattered field from a
 known object, the same principle can be used for illusion optics
 \cite{Lai:2009:IOO,Zheng:2010:EOC}.
\end{remark}

%%%%%%%%%%%%%%%%%%%%%%%%%%%%%%%%%%%%%%%%%%%%%%%%%%%%%%%%%%%%%%%%%%%%%%%%
\subsection{Active exterior cloak}
The main idea here is to achieve a similar effect to the Green's
identity cloak but without completely surrounding the cloaked region by
monopoles and dipoles on $\partial D$. We ``open the cloak'' by
replacing the single and double layer potential on each face $\partial
D_l$ of $\partial D$ by a corresponding multipolar device located at
some point $\Bx_l$. Each device produces a linear combination of
outgoing spherical waves of the form
\begin{equation}
 u_d(\Bx) = 
 \sum_{l=1}^{n_{dev}} \sum_{n=0}^\infty \sum_{m=-n}^n b_{l,n,m}
 V_n^m(\Bx-\Bx_l),
 \label{eq:udev}
\end{equation}
where $n_{dev}$ is the number of devices (or faces of $\partial D$) and
$V_n^m(\Bx)$ is a radiating, spherical wave defined for $\Bx \neq 0$ by 
\[
 V_n^m(\Bx) = h_n^{(1)}(k|\Bx|) Y_n^m(\hat{\Bx}).
\]
Here $h_n^{(1)}(t)$ is a spherical Hankel function of the first kind
(see e.g. \cite[\S 10.47]{olver:2010:NHM}) and $Y_n^m(\hat{\Bx})$ is a spherical
harmonic evaluated at the point $\hat{\Bx} \equiv \Bx/\abs{\Bx}$ of the
unit sphere $S(0,1)$. In spherical coordinates, the spherical harmonics
we use are defined as in \cite[\S 2.3]{colton:1998:iae} by
\begin{equation}
 Y_n^m(\theta,\phi) = \sqrt{\frac{2n+1}{4\pi}
 \frac{(n-\abs{m})!}{(n+\abs{m})!}} P_n^{\abs{m}}(\cos \theta) e^{im\phi},
 \label{eq:ynm}
\end{equation}
where the elevation angle is $\theta \in [0,\pi]$ and the azimuth angle
is $\phi \in [0,2\pi]$. Here $P_n^{\abs{m}}(t)$ are the associated Legendre
functions 
\[
 P_n^m(t) = (1-t^2)^{m/2} \frac{d^m P_n(t)}{dt^m}, 
\]
defined for $n=0,1,2,\ldots$ and $m=0,1,\ldots,n$ in terms of the
Legendre polynomials $P_n$ of degree $n$ with normalization $P_n(1)=1$.
The definition \eqref{eq:ynm} ensures that the spherical harmonics
$Y_n^m$ have unit $L^2(S(0,1))$ norm.

The main tool to replace the fields generated by a face is the addition 
formula (see e.g. Theorem 2.10 in \cite{colton:1998:iae})
\begin{equation}
 G(\Bx,\By) = ik \sum_{n=0}^\infty \sum_{m=-n}^n V_n^m(\Bx)
 \conj{U_n^m(\By)}
 \label{eq:mpe}
\end{equation}
which means we can mimic a point source located at $\By$ by a multipolar
source located at the origin. The coefficients in the multipolar
expansion are values of entire spherical waves
\[
 U_n^m(\Bx) = j_n(k|\Bx|) Y_n^m(\hat{\Bx}),
\]
where $j_n(t)$ are spherical Bessel functions \cite[\S 10.47]{olver:2010:NHM}. The
series in the multipolar expansion \eqref{eq:mpe} converges uniformly on
compact sets of $|\Bx| > |\By|$.

We are now ready to state the main result of this section.
\begin{theorem}
 \label{thm:conv}
 Multipolar sources located at the points $\Bx_l \notin \partial D$,
 $l=1,\ldots,n_{dev}$ can be used to reproduce the Green's formula cloak
 outside of the region
 \[
  A = \bigcup_{l=1}^{n_{dev}} B\M{\Bx_l, \sup_{\By \in \partial D_l} |\By -
  \Bx_l| },
 \]
 where $B(\Bx,r)$ is the closed ball of radius $r$ centered at $\Bx$.
 The coefficients $b_{l,n,m}$ in \eqref{eq:udev} such that
 $u_d^{(ext)}(\Bx)  = u_d(\Bx)$ for $\Bx \notin A$ are
 \begin{equation}
  \begin{aligned}
  b_{l,n,m} = ik \int_{\partial D_l} \D S_\By 
  \Bigl\lbrace&
  (-\Bn(\By)\cdot \nabla_\By u_i(\By)) \conj{U_n^m(\By - \Bx_l)}\\
  &+ u_i(\By) \Bn(\By) \cdot \nabla_\By \conj{U_n^m(\By - \Bx_l)}
  \Bigr\rbrace.
  \end{aligned}
  \label{eq:blnm}
 \end{equation}
 Moreover the convergence of \eqref{eq:udev} is uniform on compact sets
 outside $A$.
\end{theorem}
\begin{proof}
Splitting the integral in \eqref{eq:green} into integrals over each of
the faces $\partial D_l$ of the polyhedron $\partial D$ and applying the
addition theorem \eqref{eq:mpe} with center at the corresponding $\Bx_l$
we obtain:
\begin{equation}
 \begin{aligned}
 u_d(\Bx) = ik \sum_{l=1}^{n_{dev}}
 \int_{\partial D_l} \D S_\By &(-\Bn(\By)\cdot \nabla_\By u_i(\By))
  \sum_{n=0}^\infty \sum_{m=-n}^n V_n^m(\Bx-\Bx_l) 
   \conj{ U_n^m(\By-\Bx_l)}\\
  &+ u_i(\By) \Bn(\By) \cdot \nabla_\By 
  \sum_{n=0}^\infty \sum_{m=-n}^n V_n^m(\Bx-\Bx_l) 
   \conj{ U_n^m(\By-\Bx_l)}.
 \end{aligned}
 \label{eq:green2}
\end{equation}
The result \eqref{eq:blnm} follows for $\Bx \notin A$ by switching the
order of the sum and the integral in \eqref{eq:green2}.  For the first
term in the integrand of \eqref{eq:green2}, this switch is justified by
the uniform convergence of the series \eqref{eq:mpe} (for all devices) in
compact sets outside of $A$.

For the second term in the integrand of \eqref{eq:green2}, we shall show
that the series converges uniformly on compact sets outside $A$, so it
is also valid to switch the integral and the series in
\eqref{eq:green2}. To see the uniform convergence, it is useful to split
the products $V_n^m(\Bx-\Bx_l) \nabla_\By \conj{U_n^m(\By-\Bx_l)}$ into
two terms corresponding to the two terms in the gradient
\begin{equation}
 \begin{aligned}
  \nabla_\By U_n^m(\By) &= k \hat{\By} j_n'(k\abs{\By}) Y_n^m(\hat{\By})
  &&+ j_n(k\abs{\By})\frac{\abs{\By}^2 I - \By \By^T}{\abs{\By}^3} (\nabla
  Y_n^m) (\hat{\By})\\
  & = g^{(1)}_{n,m}(\By) &&+g^{(2)}_{n,m}(\By),
  \end{aligned}
  \label{eq:gradunm}
\end{equation}
where $I$ is the $3 \times 3$ identity matrix.

For the series involving the first term in the gradient
\eqref{eq:gradunm} we bound with the triangle and Cauchy-Schwarz
inequalities:
\[
 \begin{aligned}
 &\abs{\sum_{m=-n}^n V_n^m(\Bx-\Bx_l) \conj{g^{(1)}_{n,m}(\By-\Bx_l)}}\\
 &\leq 
 k \abs{h_n^{(1)}(k\abs{\Bx-\Bx_l}) j'_n(k\abs{\By-\Bx_l})}
 \M{ \sum_{m=-n}^n \abs{ Y_n^m ( \hat{\Bx-\Bx_l} ) }^2 }^{\frac{1}{2}}
 \M{ \sum_{m=-n}^n \abs{ Y_n^m ( \hat{\By-\Bx_l} ) }^2 }^{\frac{1}{2}}.
 \end{aligned}
\]
Using the summation theorem for spherical harmonics
(see e.g. Theorem 2.8 in \cite{colton:1998:iae}) 
\begin{equation}
 \sum_{m=-n}^n \abs{Y_n^m(\hat{\By})}^2 = \frac{2n+1}{4\pi}, \text{for
 any $\hat{\By} \in S(0,1)$ and $n = 0,1,\ldots$,}
 \label{eq:ynmsum}
\end{equation}
we get the estimate:
\begin{equation}
 \begin{aligned}
& \abs{\sum_{m=-n}^n V_n^m(\Bx-\Bx_l) \conj{g^{(1)}_{n,m}(\By-\Bx_l)}}\\
  & \leq k\abs{h_n^{(1)}(k\abs{\Bx-\Bx_l}) j'_n(k\abs{\By-\Bx_l})}
 \frac{2n+1}{4\pi} \\
 & = \CO\M{\frac{\abs{\By-\Bx_l}^{n-1}}{\abs{\Bx-\Bx_l}^{n+1}}},
 \end{aligned}
 \label{eq:estimate1}
\end{equation}
for large $n\to \infty$. The last equality comes from
the asymptotic expansion of Bessel functions for fixed $t>0$ and large
order $n$, (see e.g.  \cite[\S 10.19]{olver:2010:NHM})
\[
 \abs{j_n'(t)} = \CO(t^{n-1}) 
~~\text{and}~~ 
 \abs{h_n^{(1)}(t)} = \CO(t^{-n-1}).
\]

For the series involving the second term in the gradient
\eqref{eq:gradunm} we bound the sums
\[
\begin{aligned}
 &\abs{\sum_{m=-n}^n V_n^m(\Bx-\Bx_l) \conj{g^{(2)}_{n,m}(\By-\Bx_l)} }\\
 &\leq   
 2\abs{ h_n^{(1)}(k\abs{\Bx-\Bx_l}) \frac{j_n(k\abs{\By-\Bx_l})}{\abs{\By-\Bx_l}} }
 \M{ \sum_{m=-n}^n \abs{ Y_n^m ( \hat{\Bx-\Bx_l} ) }^2 }^{\frac{1}{2}} 
 \M{ \sum_{m=-n}^n \abs{ (\nabla Y_n^m) ( \hat{\By-\Bx_l} ) }^2 }^{\frac{1}{2}}
 .
\end{aligned}
\]
Using the summation theorem for spherical harmonics \eqref{eq:ynmsum} and
their gradients (see e.g. (6.56) in \cite{colton:1998:iae}),
\begin{equation}
 \sum_{m=-n}^n \abs{(\nabla Y_n^m)(\hat{\By})}^2 = \frac{n (n+1)
 (2n+1)}{4\pi}, ~~ \text{for any $\hat{\By} \in S(0,1)$},
\end{equation}
we get the asymptotic
\begin{equation}
 \begin{aligned}
 &\abs{\sum_{m=-n}^n V_n^m(\Bx-\Bx_l) \conj{g^{(2)}_{n,m}(\By-\Bx_l)} }\\
 &\leq   
 2\abs{ h_n^{(1)}(k\abs{\Bx-\Bx_l}) \frac{j_n(k\abs{\By-\Bx_l})}{\abs{\By-\Bx_l}} }
 \M{\frac{2n+1}{4\pi}}^{\frac{1}{2}}
 \M{\frac{n(n+1)(2n+1)}{4\pi}}^{\frac{1}{2}}\\
 & = \CO\M{\frac{\abs{\By-\Bx_l}^{n-1}}{\abs{\Bx-\Bx_l}^{n+1}}}.
 \end{aligned}
 \label{eq:estimate2}
\end{equation}
Here we have used that for $t>0$ fixed and as $n\to \infty$, (see e.g.
\cite[\S 10.19]{olver:2010:NHM}) 
\[
 \abs{j_n(t)} = \CO(t^{n}) 
~~\text{and}~~ 
 \abs{h_n^{(1)}(t)} = \CO(t^{-n-1}).
\]

The estimates \eqref{eq:estimate1} and \eqref{eq:estimate2} give
uniformly convergent majorants for the series in the second term of
\eqref{eq:green2}, since when $\Bx \notin A$ we have $\abs{\By-\Bx_l} <
\abs{\Bx-\Bx_l}$, for $l=1,\ldots,n_{dev}$. The proof is now completed.
\qed
\end{proof}

%%%%%%%%%%%%%%%%%%%%%%%%%%%%%%%%%%%%%%%%%%%%%%%%%%%%%%%%%%%%%%%%%%%%%%%%
\subsection{A family of exterior cloaks with four devices}
\label{sec:family}

%%%%%%%%%%%%%%%%%%%%%%%%%%%%%%%%%%%%%%%%%%%%%%%%%%%%%%%%%%%%%%%%%%%%%%%%
\begin{figure}
 \begin{center}
 \begin{tabular}{cc}
  \includegraphics[width=0.4\textwidth]{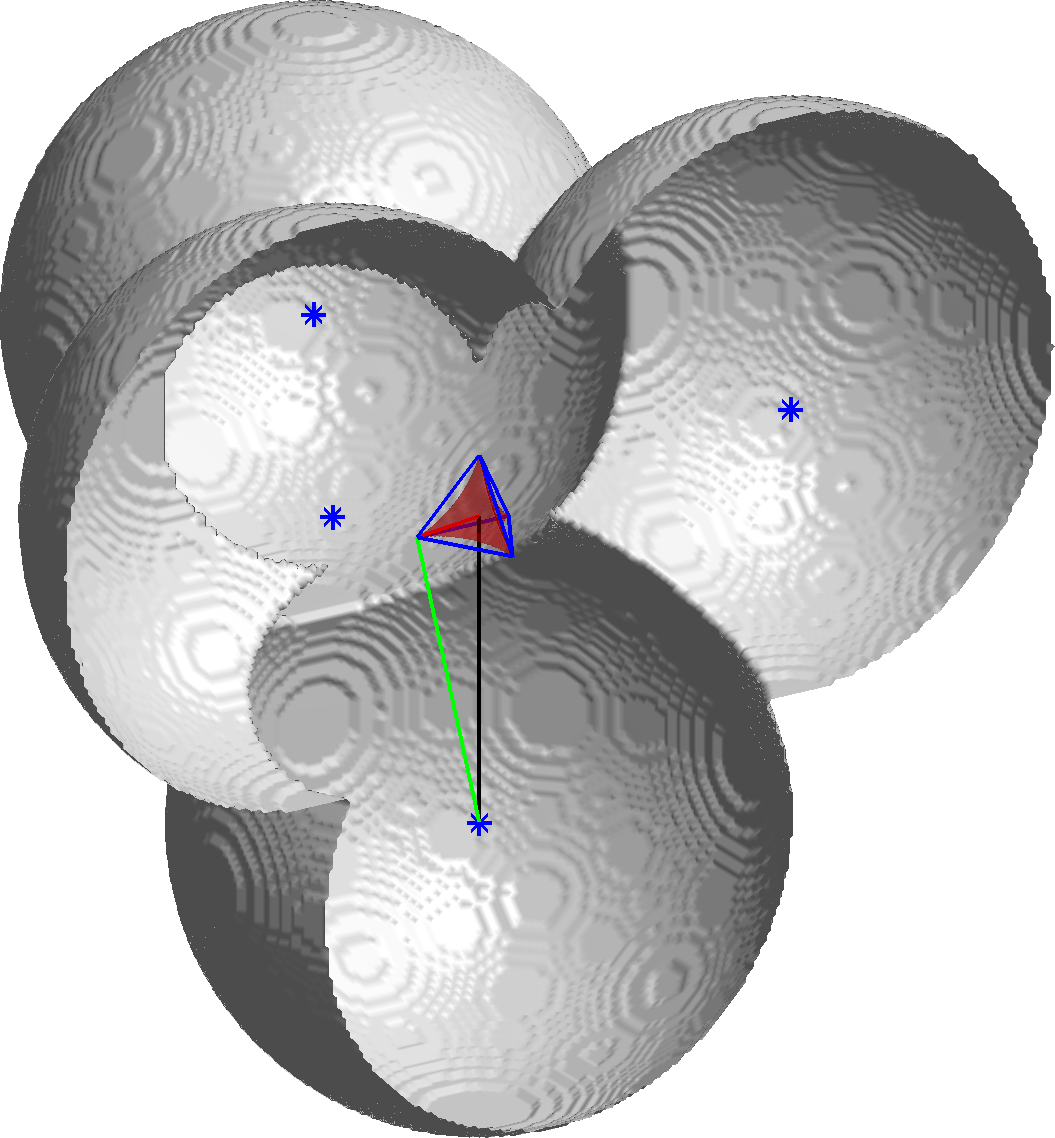} &
  \includegraphics[width=0.4\textwidth]{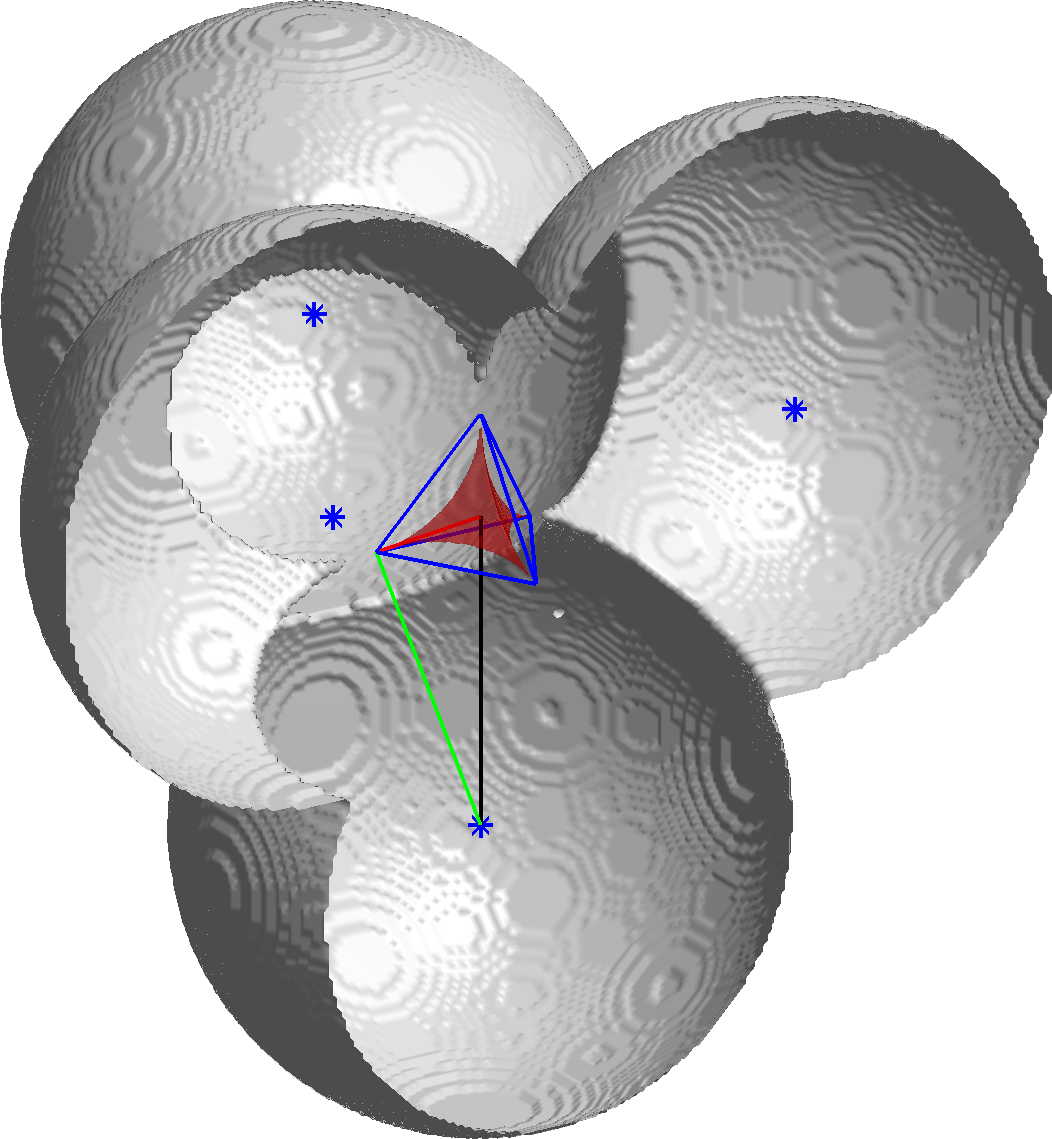}\\
  (a) suboptimal, $\sigma = \delta/5$ & (b) optimal, $\sigma = \delta/3$
 \end{tabular}
 \end{center}
 \caption{The configuration for the tetrahedron based cloak of
 \secref{sec:family}. The distance in red is the radius $\sigma$ of the
 circumsphere to the tetrahedron $D$. The distance in black is the
 distance $\delta$ from the origin to a device. The distance
 $r(\sigma,\delta)$ (in green) is the distance from a device to the
 closest vertex of $D$. The exterior surface of the region 
 $A$ of \thmref{thm:conv} is in grey and has been cut to reveal the cloaked region
 $D\setminus A$ in red. The four devices are shown with stars.}
 \label{fig:tetra}
\end{figure}

Nothing in \thmref{thm:conv} guarantees that the cloaked region
$D\setminus A$ is non-empty. We show here how to construct a family of
cloaks with non-empty $D\setminus A$ based on Green's identities
applied to a regular tetrahedron $D$. We also determine what is the
position of the devices that gives the largest cloaked region within
this family.

Consider a regular tetrahedron with circumsphere $S(0,\sigma)$ and
vertices $\Ba_1,\ldots,\Ba_4$. We locate the devices
$\Bx_1,\ldots,\Bx_4$ on $S(0,\delta)$, with $\delta > \sigma$, such that
$\Bx_l$ replaces the face opposite to vertex $\Ba_l$, that is $\Bx_l$ and
$\Ba_l$ are on opposite sides of the plane formed by the face of the
tetrahedron not containing $\Ba_l$. For simplicity we also require that
$\Bx_l - \Ba_l$ is normal to this plane. The configuration is sketched
in \figref{fig:tetra}. Simple geometric arguments show that the radii of
the balls that define the region $A$ are all equal to
\begin{equation}
 r(\sigma,\delta) = \M{\M{\sigma - \frac{\delta}{3}}^2 + \frac{8}{9}
 \delta^2 }^{\frac{1}{2}}.
\end{equation}
Moreover the radius of the largest sphere fitting inside the cloaked
region is
\begin{equation}
 r_{\text{eff}}(\sigma,\delta) = \delta - r(\sigma,\delta).
\end{equation}
For fixed $\delta$, the largest possible cloaked region is obtained when
$\sigma = \delta/3$ which corresponds to the case when every triplet of
balls in the definition of region $A$ touch at a vertex $\Ba_l$ of the
tetrahedron.  Thus for fixed $\delta$, the largest sphere we can fit
inside the cloaked region has radius,
\begin{equation}
 r_{\text{eff}}^* = \M{1- \frac{2\sqrt{2}}{3} }\delta \approx 0.057 \delta.
 \label{eq:reff}
\end{equation}

%%%%%%%%%%%%%%%%%%%%%%%%%%%%%%%%%%%%%%%%%%%%%%%%%%%%%%%%%%%%%%%%%%%%%%%%
\subsection{Numerical experiments}
We report in \figref{fig:fields} simulations of this cloaking method
with the setup described in \secref{sec:family}. The incident field we
take is the plane wave $u_i(\Bx) = \exp[i k \hat{\Bk} \cdot \Bx]$ with
direction  vector $\hat{\Bk} = [1,1,1]/\sqrt{3}$. We first compute the
device field of \thmref{thm:conv} by truncating the sum in $n$ of
\eqref{eq:udev} to $n\leq N$. Throughout our numerical experiments we
determine $N$ with the heuristic (found by numerical experimentation)
\begin{equation} 
 N(\delta) = \ceil{1.5 k \delta }, 
 \label{eq:heuristic}
\end{equation}
where $\ceil{x}$ is the smallest integer larger than or equal to $x$.
The integrals in \eqref{eq:blnm} were evaluated with a simple quadrature
rule that is exact for piecewise linear functions on a uniform
triangulation of the faces of the tetrahedron $D$, we chose the number
of quadrature points so that there are at least eight points per
wavelength. The scattered field by a ball was computed by first
evaluating the incident field (or device field depending on the case) on
a grid with equal number of points in $\phi$ and $\theta$ and then
finding its first few spherical harmonic decomposition coefficients
using the sampling theorem \cite{driscoll:1994:CFT}.

As can be seen in the first  row of \figref{fig:fields} the device field
$u_d$ is virtually zero far from $A$ while being close to the incident
field in the cloaked region $D\setminus A$. In the second and third
rows of \figref{fig:fields} we display the total field in the presence
of a sound-soft (homogeneous Dirichlet boundary condition) ball centered
at the origin and of radius $3 r_{\text{eff}}^*(\delta)$ (i.e. a larger
scatterer than what we expected from \secref{sec:family}). The scattered
field from the ball reveals the ball's position when the devices are
inactive (third row). The scattered field is essentially suppressed when
the cloaking devices are active (second row), as the field is
indistinguishable from a plane wave far from $A$.

Since as $t\to 0$, $h_n^{(1)}(t) = \CO(t^{-n-1})$ (see e.g.  \cite[\S
10.52]{olver:2010:NHM}), we expect the device field $u_d$ to blow up as
we get close to the device locations $\Bx_l$. This blow up corresponds
to the ``urchins'' in the first and second rows of \figref{fig:fields}
where even with the truncation of the series \eqref{eq:udev}, we
observe very large wave amplitudes which would be hard to realize in
practice. Fortunately we can enclose the regions with very large fields
by a surface and apply Green's formula \eqref{eq:green} to replace these
large fields by (hopefully) more manageable single and double layer
potentials on the surface of some ``extended'' cloaking devices. 

We illustrate these ``extended'' devices in \figref{fig:vol} where we
display the level sets where the device field amplitude is 5 (or 100)
times the amplitude of the incident field. At least for the particular
configuration ($\delta = 6\lambda$) considered in \figref{fig:vol},
these surfaces resemble spheres surrounding each device location
$\Bx_l$. The ``extended'' devices still leave the cloaked region (in red
in \figref{fig:vol}) communicating (connected) with the background medium. This is
why we call our cloaking method ``exterior cloaking''.

We also consider the extended devices for larger values of $\delta$ in
\figref{fig:map}. Here we look at the cross-section of the extended
devices on $S(0,\sigma)$, which in the construction of
\secref{sec:family} is the circumsphere to the tetrahedron $D$. In the
optimal case $\delta = 3\sigma$, the predicted cloaked region
$D\setminus A$ and the exterior $\real^3 \setminus A$ meet on
$S(0,\sigma)$ at the vertices of the tetrahedron $D$. We see that the
extended devices (in black in \figref{fig:map}) grow as $\delta$
increases, and leave gorges communicating the cloaked region with the
exterior. The centers of the gorges appear to agree with the vertices of
the tetrahedron $D$. The percentage area of $S(0,\sigma)$ that is not
covered by the cross-section of the extended devices on $S(0,\sigma)$ is
also quantified in \figref{fig:perf}(b). Since the relative area of the
openings appears to decrease monotonically with $\delta/\lambda$,
\figref{fig:perf}(b) suggests the gorges close for large enough
$\delta/\lambda$. Further investigation is needed to find out whether
the shrinking openings in the cloak is due to our choice of $N$ with
heuristic \eqref{eq:heuristic}.

Finally we give in \figref{fig:perf}(a) quantitative measures of the
cloak performance for different values of $\delta$. These measures show
that the device field is close to minus the incident field inside the
cloaked region and that it is very small outside of the cloaked region.

%%%%%%%%%%%%%%%%%%%%%%%%%%%%%%%%%%%%%%%%%%%%%%%%%%%%%%%%%%%%%%%%%%%%%%%%
\begin{figure}
\begin{center}
% macros used for plotting
\newcommand{\onerow}[1]{
 \includegraphics[width=0.18\textwidth]{#1_sl1} & 
 \includegraphics[width=0.18\textwidth]{#1_sl2} & 
 \includegraphics[width=0.18\textwidth]{#1_sl3} & 
 \includegraphics[width=0.18\textwidth]{#1_sl4} & 
 \includegraphics[width=0.18\textwidth]{#1_sl5}}

\begin{tabular}{c@{}c@{\hspace{0.01\textwidth}}c@{\hspace{0.01\textwidth}}c@{\hspace{0.01\textwidth}}c@{\hspace{0.01\textwidth}}c}
 & $z=-2\sigma$ & $z=-\sigma$ & $z=0$ & $z=\sigma$ & $z=2\sigma$\\
 \rlab{$u_d$}{1.5em}                  & \onerow{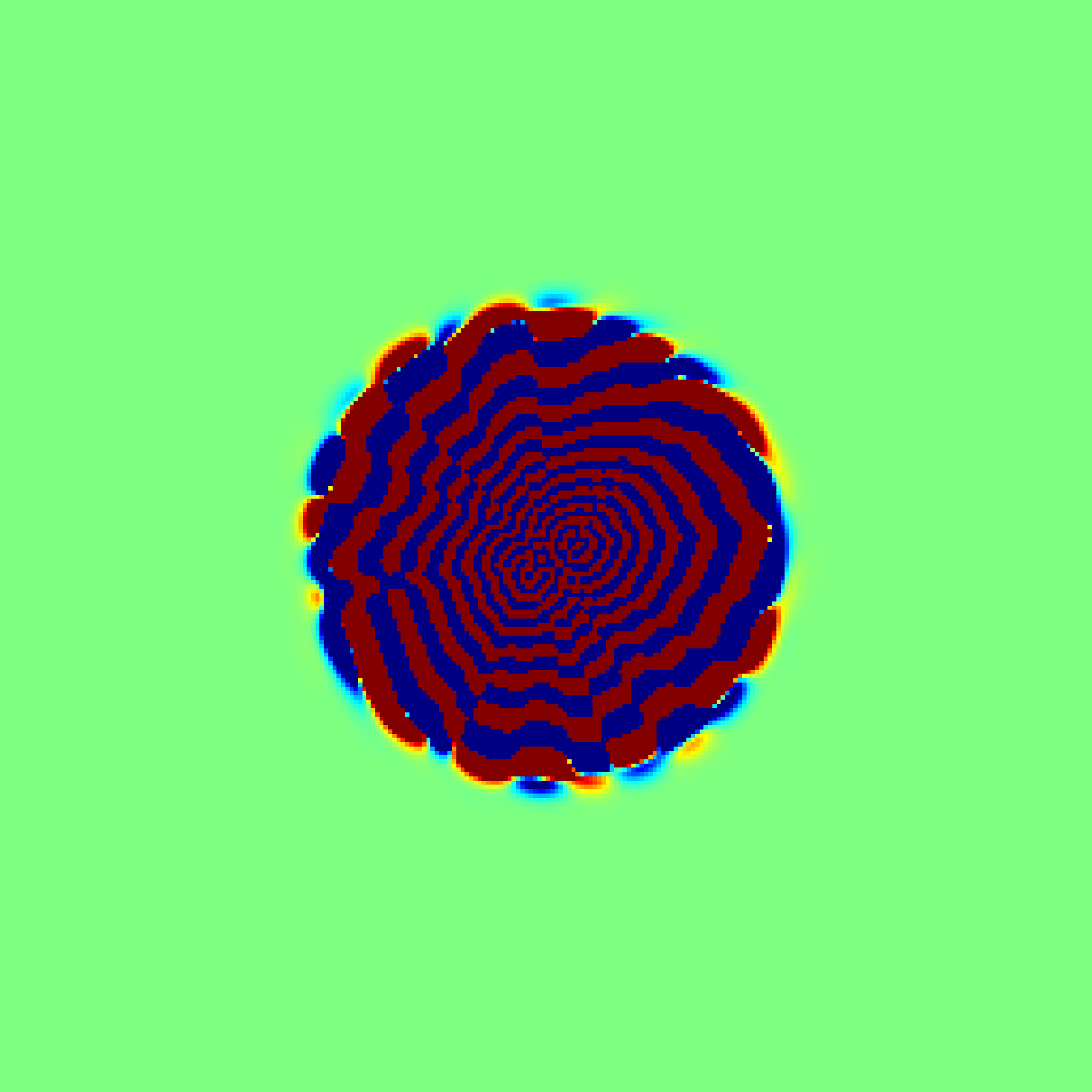}\\
 \rlab{$u_{tot}$ (active)}{0.0em}     & \onerow{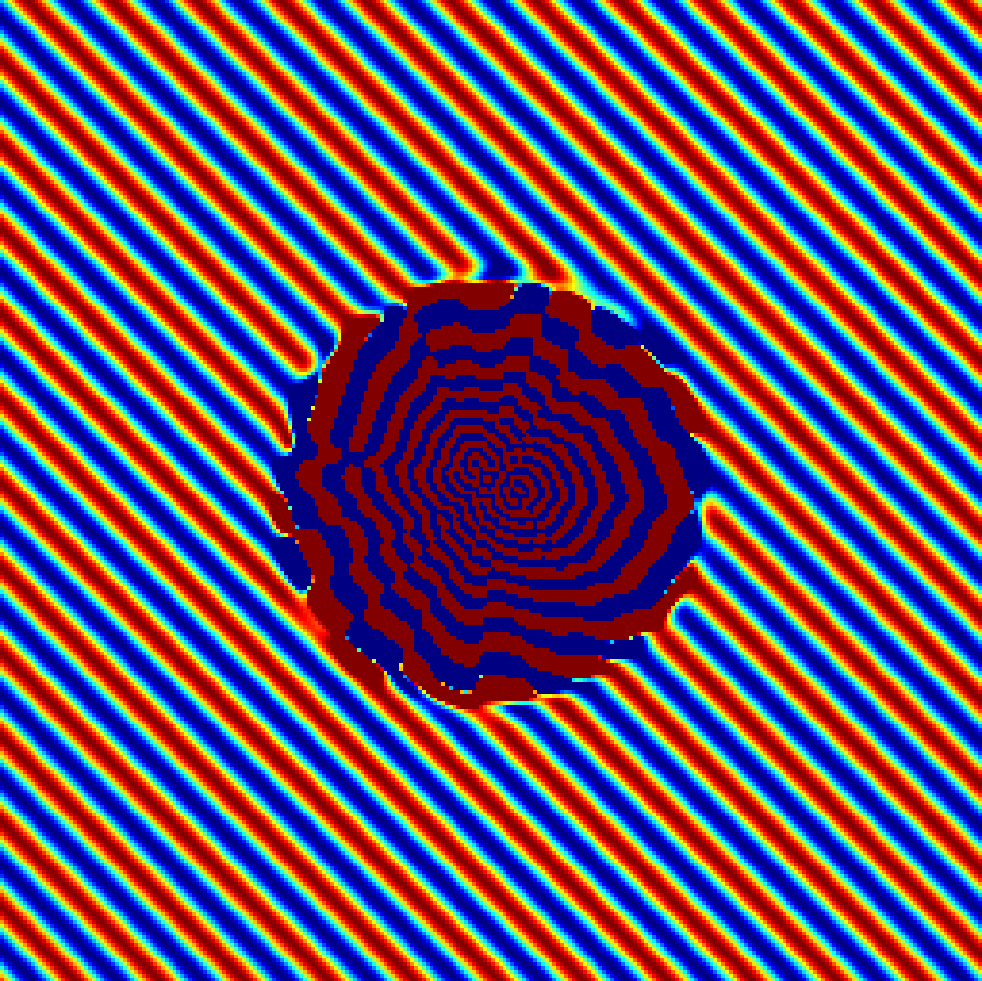}\\
 \rlab{$u_{tot}$ (inactive)}{0.0em}   & \onerow{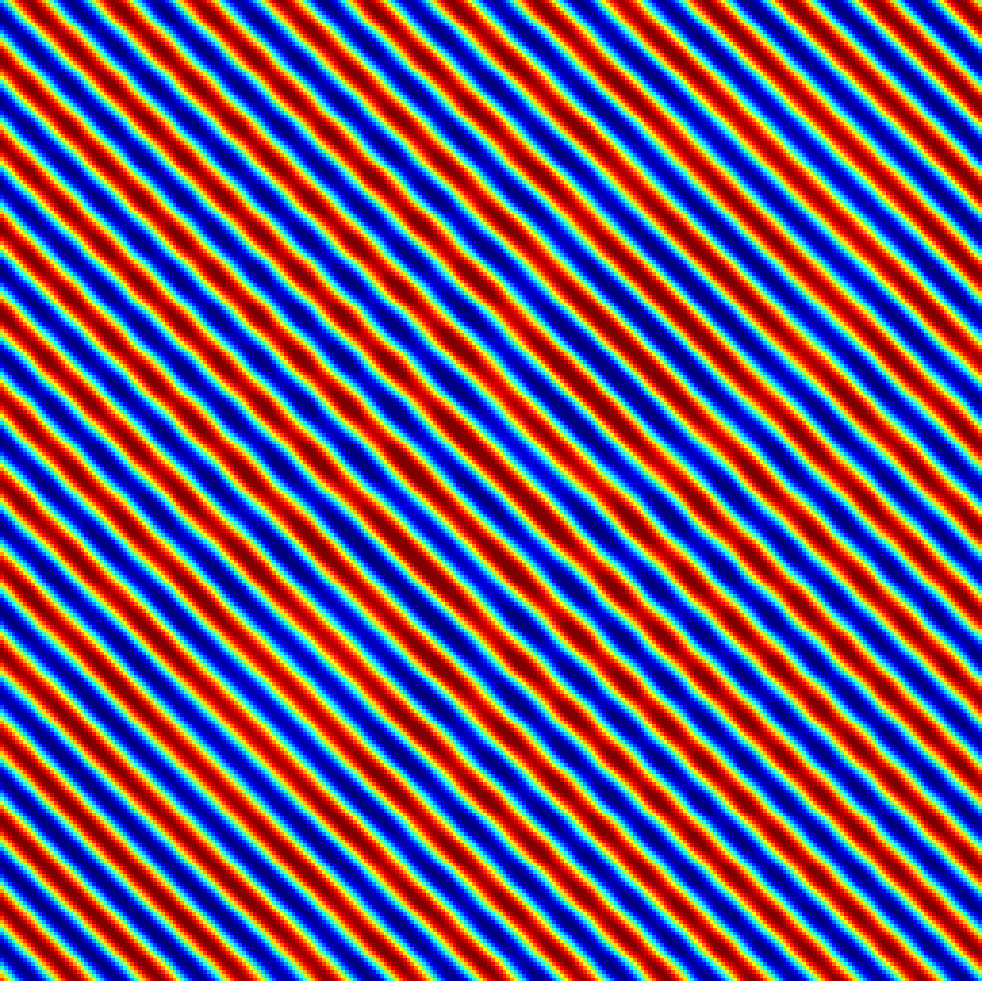}
\end{tabular}
\end{center}
\caption{Constant $z$ slices of the real part of different fields, for
the optimal case $\delta=3\sigma$ and with $\delta=6\lambda$. The first
row shows the device field $u_d$ which is close to zero
far from the devices and close to $-u_i$ in a small region close to the
origin.  The second and third rows show the total field when the devices
are active and inactive in the presence of a scatterer. The scatterer is
a sound-soft ball centered at the origin and of radius
$3r_{\text{eff}}^*(\delta)$. Even though this ball is not completely
contained inside the tetrahedron $D$, the scattered field is greatly
suppressed when the devices are active, making the ball harder to detect
far from the devices. The color scale is linear from -1 (dark blue) to 1
(dark red) and each box is $10\lambda \times 10 \lambda$, with the
$z-$axis at the center.}
\label{fig:fields}
\end{figure}

%%%%%%%%%%%%%%%%%%%%%%%%%%%%%%%%%%%%%%%%%%%%%%%%%%%%%%%%%%%%%%%%%%%%%%%%
\begin{figure}
\begin{center}
 \begin{tabular}{cc}
 \includegraphics[width=0.4\textwidth]{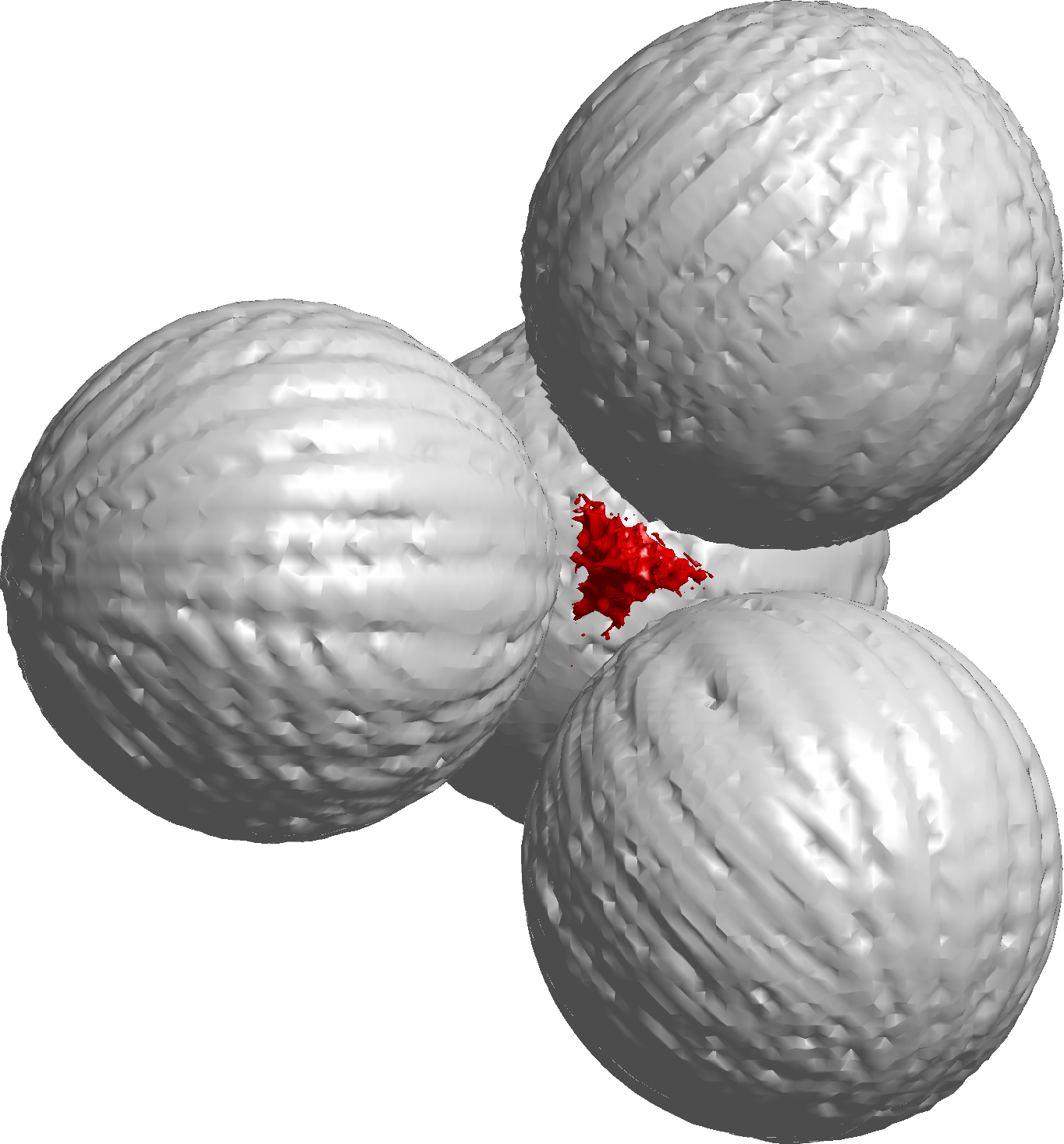} &
 \includegraphics[width=0.4\textwidth]{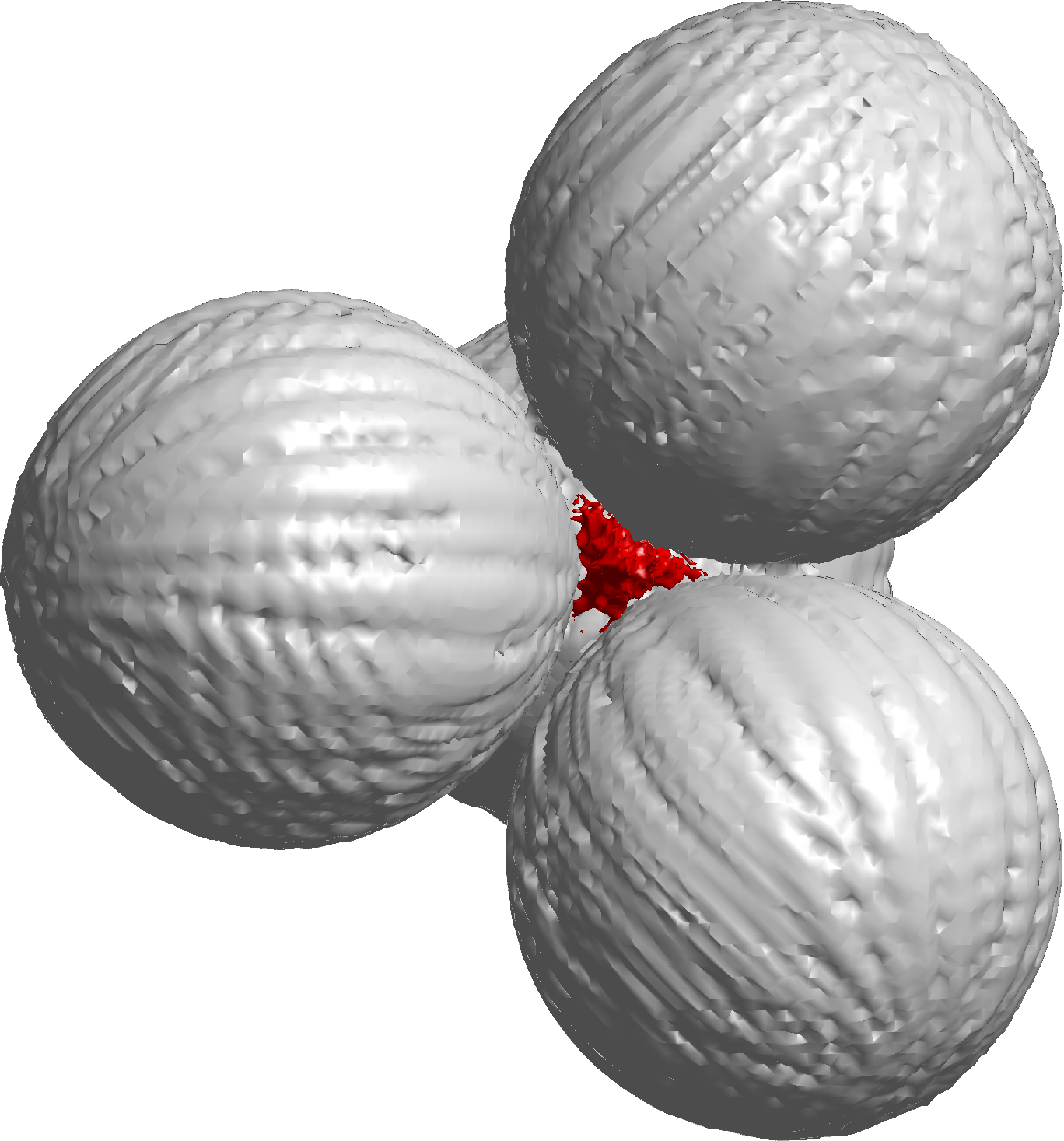}\\
 (a) $|u_d| = 100$ & (b) $|u_d| = 5$
 \end{tabular}
\end{center}
\caption{Contours of $|u_d|$ (gray) and $|u_d+u_i|=10^{-2}$ (red).
Here the vector $(0,0,1)$ is perpendicular to plane of the page. By
Green's identity it is possible to replace the large fields inside the
gray surfaces by a single and double layer potential at the gray
surfaces. These ``extended devices'' need only to generate fields that
are at most the fields on the contours that we plot and they cloak the
red region without completely surrounding it.}
\label{fig:vol}
\end{figure}

%%%%%%%%%%%%%%%%%%%%%%%%%%%%%%%%%%%%%%%%%%%%%%%%%%%%%%%%%%%%%%%%%%%%%%%%
\begin{figure}
\begin{center}
 \begin{tabular}{cc}
 $\delta=6\lambda$ & $\delta = 12 \lambda$\\
 \includegraphics[width=0.4\textwidth]{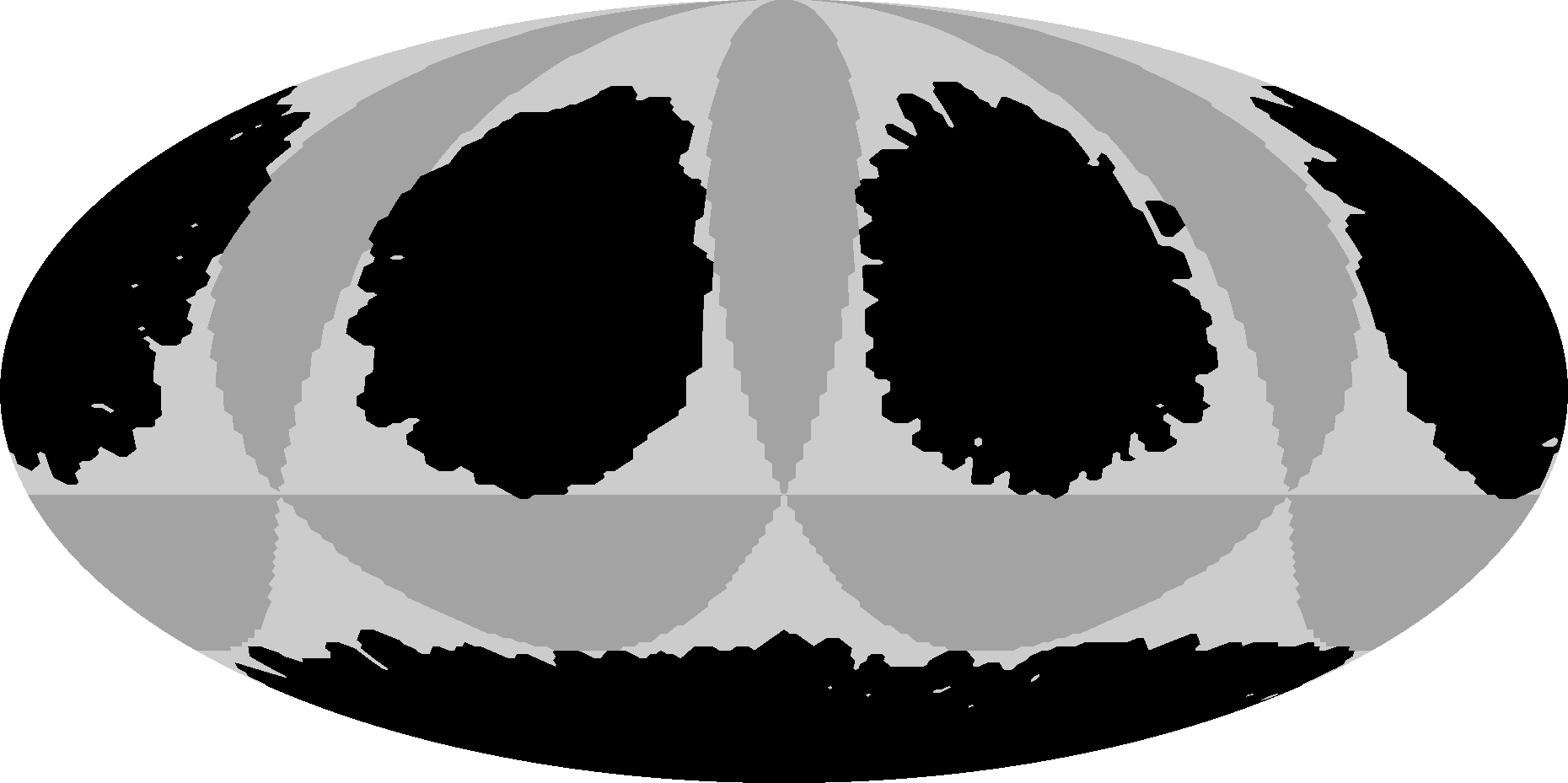} &
 \includegraphics[width=0.4\textwidth]{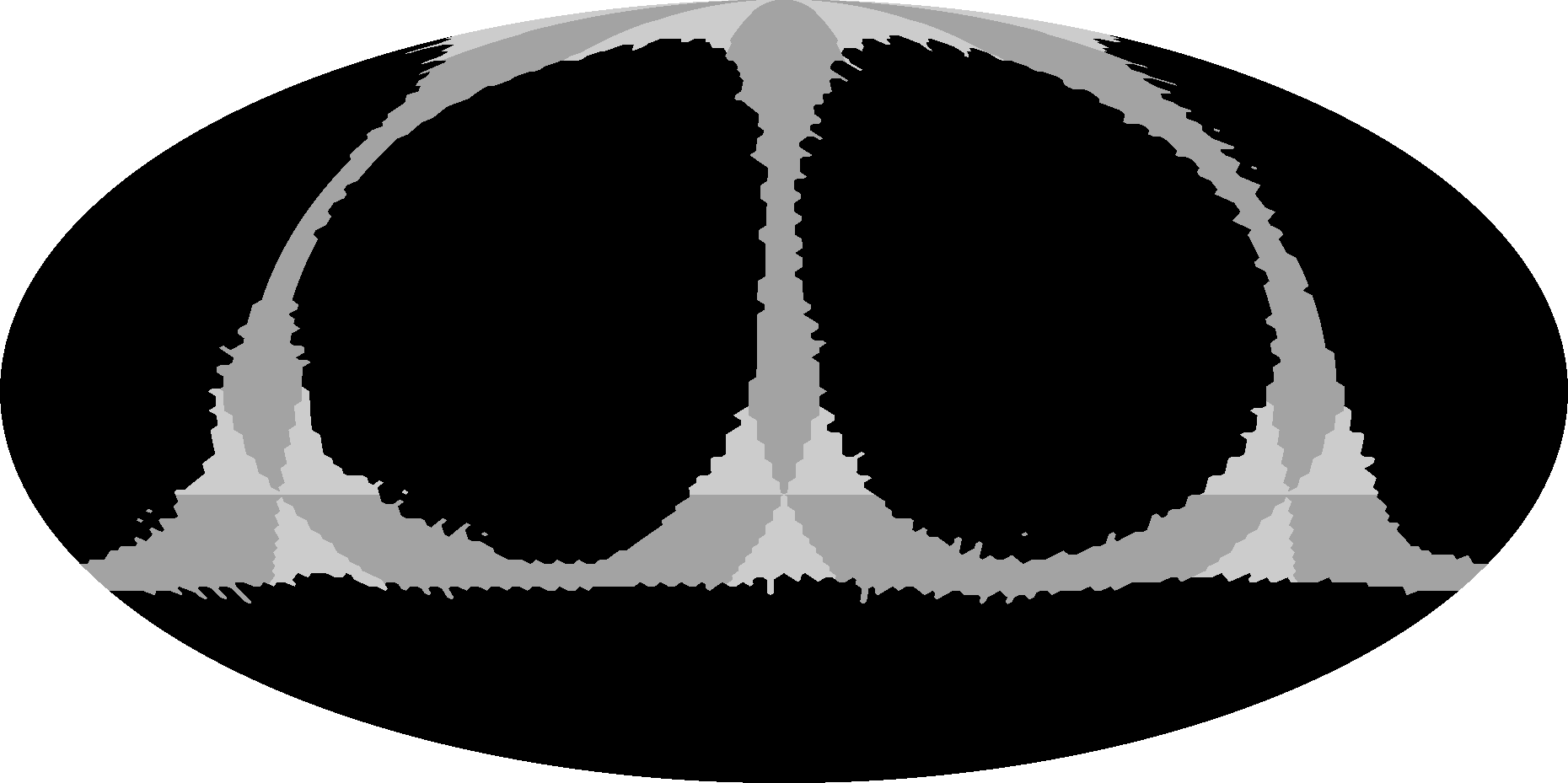}\\
 \includegraphics[width=0.4\textwidth]{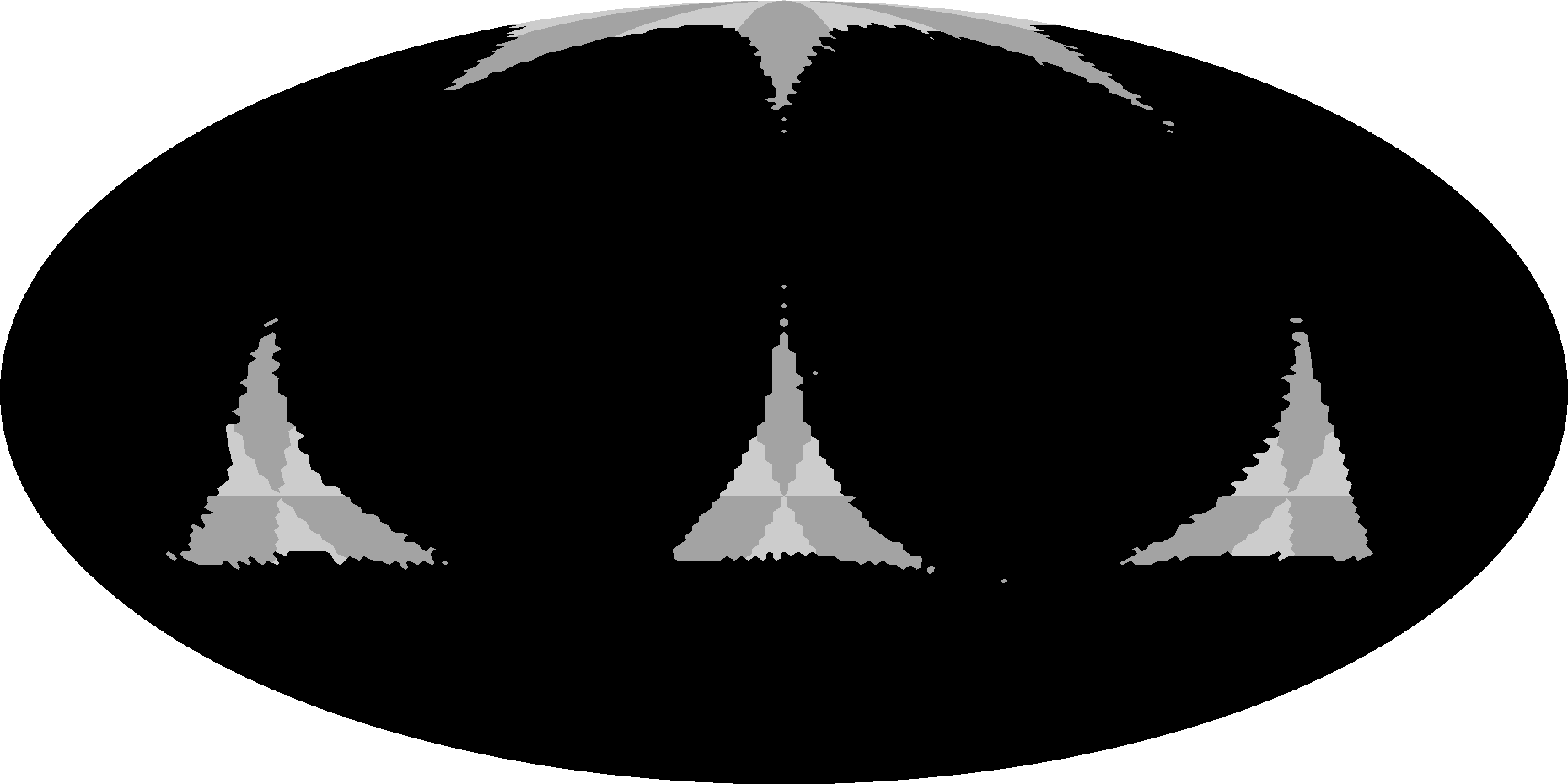} &
 \includegraphics[width=0.4\textwidth]{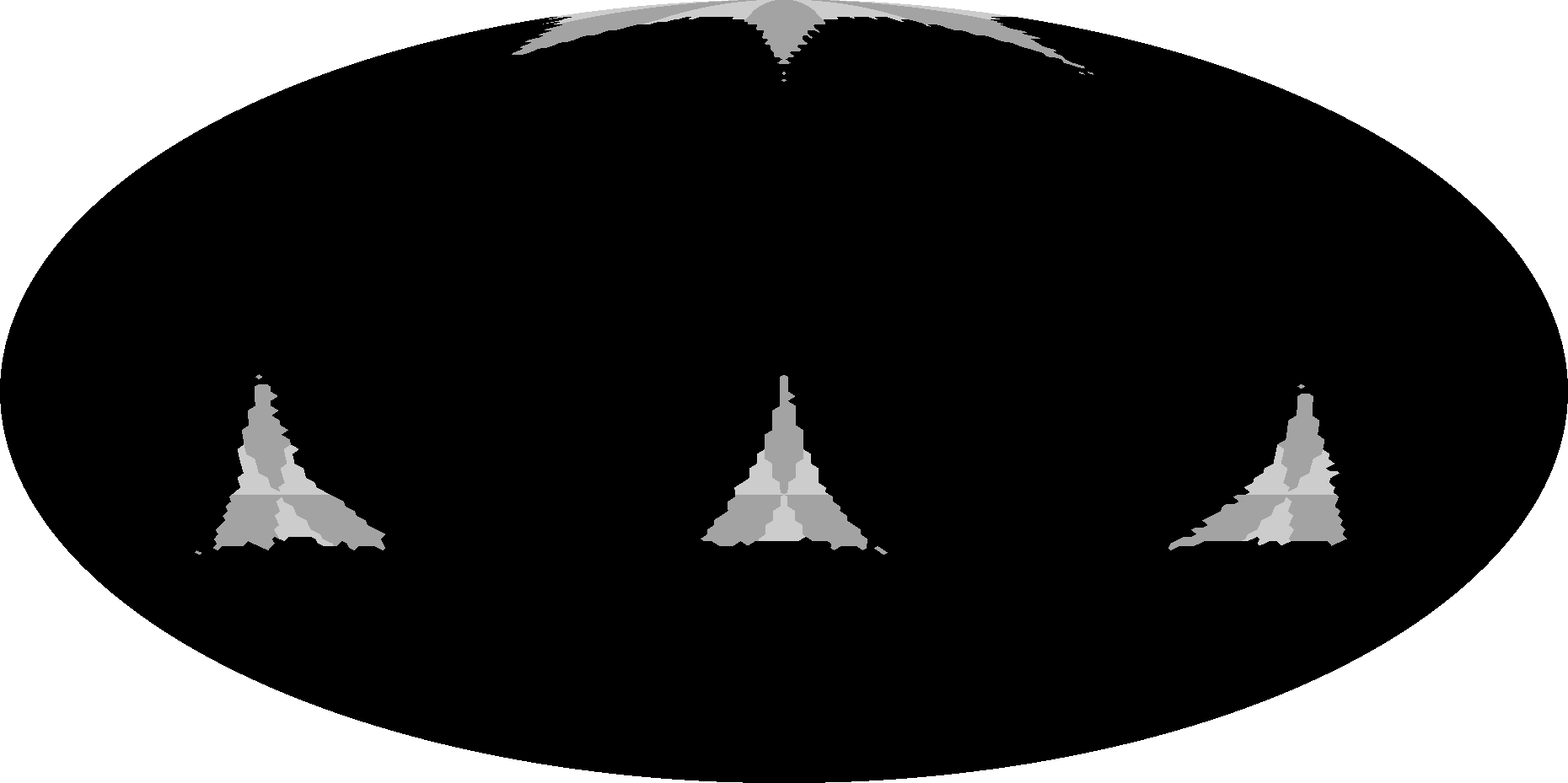}\\
 $\delta=18\lambda$ & $\delta = 24\lambda$
 \end{tabular}
\end{center}
\caption{Cross-section of level set $|u_d| \geq 10^2$ (black) and of the
region $A$ (shades of gray) on the sphere $|\Bx| = \sigma$ for  the
optimal $\sigma=\delta/3$. Here we used the equal area Mollweide projection
(see e.g.  \cite{feeman:2002:PTE}). In the optimal case, each triplet
out of the four balls forming $A$ meets at a single point which is a
vertex of the tetrahedron $D$. Note that for the cases in the first row
there are four distinct extended devices. The leftmost and rightmost
spots correspond to one single device split in two by the projection.}
\label{fig:map}
\end{figure}

%%%%%%%%%%%%%%%%%%%%%%%%%%%%%%%%%%%%%%%%%%%%%%%%%%%%%%%%%%%%%%%%%%%%%%%%
\begin{figure}
\begin{center}
\begin{tabular}{cc}
\rlab{(percent)}{4em} \includegraphics[width=0.45\textwidth]{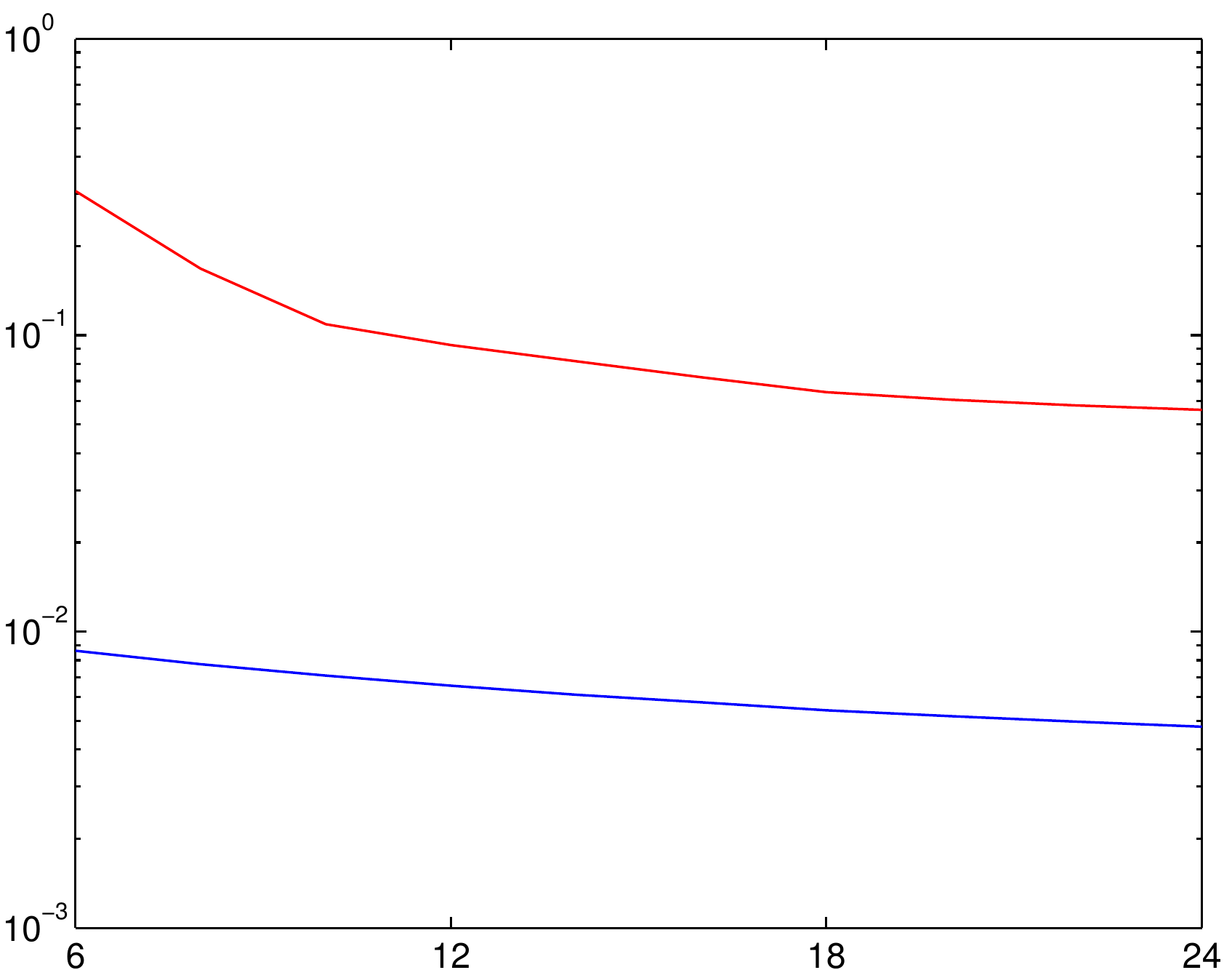} & 
\rlab{(percent)}{4em} \includegraphics[width=0.45\textwidth]{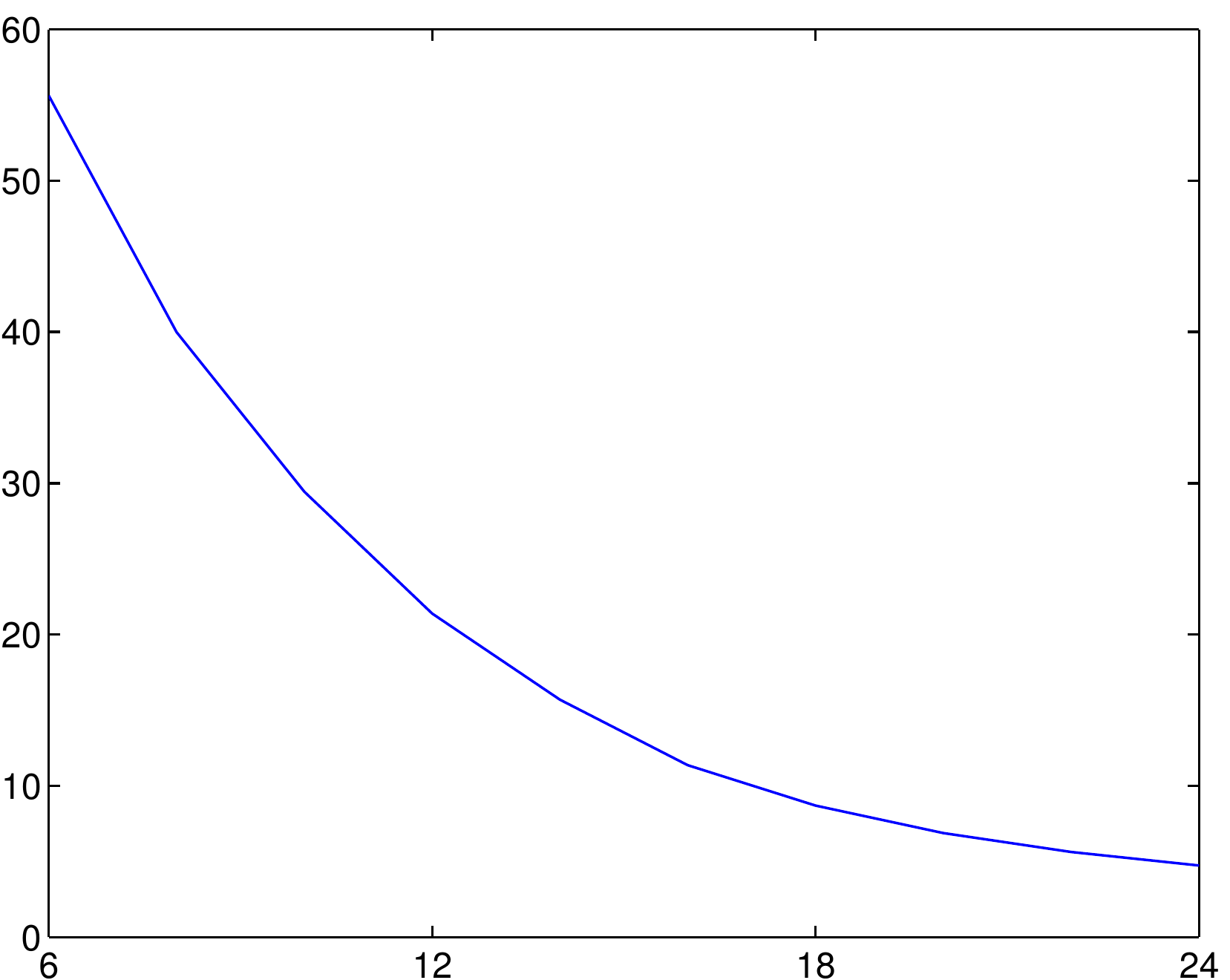} \\
 (a: $\delta/\lambda$) & (b: $\delta/\lambda$) \\
\end{tabular}
\end{center}
\caption{(a) Cloak performance. In red: $\| u_i + u_d \| / \| u_i \|$,
where the norm is the $L^2(S(0,r_{\text{eff}}^*(\delta))$ norm, which
measures how well we approximate the incident field inside the cloaked
region. In blue: $\|u_d\|/ \|u_i\|$, where the norm is the
$L^2(S(0,2\delta))$ norm, which measures how small is the device field
far away from the devices. (b) Percentage of the area outside the
cross-section of the extended devices on the sphere $S(0,\sigma=\delta/3)$ for
different values of $\delta$.}
\label{fig:perf}
\end{figure}

%%%%%%%%%%%%%%%%%%%%%%%%%%%%%%%%%%%%%%%%%%%%%%%%%%%%%%%%%%%%%%%%
\begin{acknowledgement}
 GWM is grateful for support from the University of Toulon-Var.
 GWM and DO are grateful to the National Science Foundation for support
 through grant DMS-0707978. FGV is grateful to the National Science
 Foundation for support through grant DMS-0934664.
 FGV, GWM and DO are grateful to the Mathematical Sciences Research
 Institute where parts of this manuscript were completed.
 The computations of the device and scattered fields in \secref{sec:helm} were
 facilitated by the freely available spherical harmonics library SHTOOLS
 by Mark Wieczorek, available at
 \url{http://www.ipgp.fr/~wieczor/SHTOOLS/SHTOOLS.html}. 
\end{acknowledgement}

\bibliographystyle{spbasic}
\bibliography{chbib}

\end{document}